\documentclass[aps,amsmath,amssymb,twocolumn,nofootinbib,prx,superscriptaddress]{revtex4-2}

\usepackage{graphicx}
\usepackage{dcolumn}
\usepackage{bm}
\usepackage{epsfig,hyperref,amsthm}
\usepackage{mathtools}
\usepackage{color}
\usepackage{colordvi}
\usepackage[dvipsnames]{xcolor}
\usepackage{relsize}
\usepackage{accents}
\usepackage{wasysym}  
\usepackage{xypic}
\usepackage{times}

\hypersetup{
	colorlinks=true,
	linkcolor=CitingColor,
	citecolor=CitingColor,
	urlcolor=CitingColor
}

\DeclareMathAlphabet\mathbfcal{OMS}{cmsy}{b}{n}
\newcommand{\ket}[1]{\ensuremath{|#1\rangle}}
\newcommand{\bra}[1]{\ensuremath{\langle #1|}}
\newcommand{\braket}[2]{\langle #1|#2\rangle}
\newcommand{\proj}[1]{\ket{#1}\bra{#1}}
\newcommand{\be}{\begin{equation}}
\newcommand{\ee}{\end{equation}}
\newcommand{\ba}{\begin{eqnarray}}
\newcommand{\ea}{\end{eqnarray}}

\newcommand{\tr}{{\rm tr}}
\newcommand{\Tr}{{\rm Tr}}

\newcommand{\id}{\mathbb{I}}

\newcommand{\wt}[1]{\widetilde{#1}}

\newcommand{\mE}{\mathcal{E}}

\newtheorem{theorem}{Theorem}

\newtheorem{alemma}{Lemma}[section]
\newtheorem{aproposition}[alemma]{Proposition}
\newtheorem{atheorem}[alemma]{Theorem}

\newtheorem{acorollary}[alemma]{Corollary}

\newtheorem{question}{Question}

\newtheorem{definition}{Definition}

\newcommand{\CY}[1]{{\color{black}#1}}
\newcommand{\CYtwo}[1]{{\color{black}#1}}
\newcommand{\toni}[1]{{\color{black}#1}}

\newcommand{\CYthree}[1]{{\color{black}#1}}

\newcommand{\HXS}{H^{\rm (XS')}_{\rm X|X'}}
\newcommand{\HXSd}{H^{\rm (XS'),\dagger}_{\rm X|X'}}
\newcommand{\ZXS}{Z^{\rm (XS')}_{\rm X|X'}}

\definecolor{nred}{rgb}{0.9,0.1,0.1}
\definecolor{nblack}{rgb}{0,0,0} 
\definecolor{nblue}{rgb}{0.2,0.2,0.8}
\definecolor{ngreen}{rgb}{0.2,0.6,0.2}
\definecolor{ublue}{rgb}{0,0,0.5}

\definecolor{pur}{rgb}{0.75,0,0.75}
\definecolor{nngrn}{rgb}{0,0.5,0.5}
\definecolor{CitingColor}{rgb}{0,0.3,1}

\newcommand{\blu}{\color{nblue}}

\begin{document}
\title{Quantum Channel Marginal Problem}

\author{Chung-Yun Hsieh}
\email{chung-yun.hsieh@icfo.eu}
\affiliation{ICFO - Institut de Ciencies Fotoniques, The Barcelona Institute of Science and Technology, 08860 Castelldefels, Spain}

\author{Matteo Lostaglio}
\affiliation{ICFO - Institut de Ciencies Fotoniques, The Barcelona Institute of Science and Technology, 08860 Castelldefels, Spain}
\affiliation{QuTech, Delft University of Technology, P.O. Box 5046, 2600 GA Delft, The Netherlands}
\affiliation{Korteweg-de Vries Institute for Mathematics and QuSoft, University of Amsterdam, The Netherlands}

\author{Antonio Ac\'in}
\affiliation{ICFO - Institut de Ciencies Fotoniques, The Barcelona Institute of Science and Technology, 08860 Castelldefels, Spain}
\affiliation{ICREA-Instituci\'o Catalana de Recerca i Estudis Avan\c{c}ats, Lluis Companys 23, 08010 Barcelona, Spain}

\date{\today}

\begin{abstract}
Given a set of local dynamics, are they compatible with a global dynamics?
We systematically formulate these questions as {\em quantum channel marginal problems}. 
These problems are strongly connected to the generalization of the no-signaling conditions to quantized inputs and outputs and can be understood as a general toolkit to study notions of quantum incompatibility. 
In fact, they include as special cases channel broadcasting, channel extendibility, measurement compatibility and state marginal problems.
After defining the notion of compatibility between global and local dynamics, we provide a solution to the channel marginal problem that takes the form of a semidefinite program.
Using this formulation, we construct channel incompatibility witnesses, discuss their operational interpretation in terms of an advantage for a state-discrimination task, prove a gap between classical and quantum dynamical marginal problems and show that the latter is irreducible to state marginal problems.
\end{abstract}

\maketitle

{\em Introduction.}--
A fundamental question in quantum mechanics is whether a given set of local states are compatible with a global one.
In other words, can the former be seen as the {\em marginals} of a global quantum state?
This kind of questions are known as {\em state marginal problems} (SMPs). 
One of the most prominent examples is the $2$-body $N$-representability problem, where one asks which $2$-body reduced density matrices can result as the marginals of a global state of $N$ particles, a problem motivated by the calculation of ground states of 2-body, usually local, Hamiltonians, see for instance~\cite{Schilling,bookchapter}. 
Because of its relevance, the SMP has been studied from many different viewpoints, for instance in the context of entanglement~\cite{tothguhne,entmarginal} or Bell non-locality detection~\cite{tura,navascues}, or by constructing efficient measurement strategies for the estimation of marginal states~\cite{tomography1,tomography2,tomography3}. 
The SMP also has a classical analogue, where incompatibility appears in the form of frustration due to loops~\cite{Wolf2003}.

Whereas the SMP is concerned with static properties encoded in states, the purpose of this work is to understand how compatibility between local and global descriptions extends at the level of dynamics. As we shall see, for the concept of local dynamics to be well-defined, the channels under consideration have to satisfy a generalization of the no-signaling condition to quantum input and outputs considered in Refs.~\cite{Beckman2001,Eggeling2002,Piani2006,Duan2016,Wang2017,Hoban2018}.
First, however, we introduce the {\em channel marginal problems} (CMPs) as a natural {\em dynamical} generalization of the SMP. 
After this, we provide necessary and sufficient conditions to solve it that can be addressed using semi-definite programming. 
There exist \CY{several} previous works that have also considered the CMP~\cite{Girard2020,Haapasalo2014,Allen2017,Haapasalo2019,Viola2015,Kaur2019}, and we shall see that our framework includes those (and more) within a unified umbrella.

It is convenient for what follows to recall the definition of the SMP, see also Figure~\ref{Fig}(a).
\begin{definition}\label{Def:SMP}
{\rm (State Marginal Problem)} Consider a global system ${\rm S}$ and a set of local states $\{\rho_{\rm X}\}_{\rm X\in\Lambda}$, where $\Lambda$ is a collection of subsystems ${\rm X}$ of ${\rm S}$ and each $\rho_{\rm X}$ is in ${\rm X}$.
Then a SMP asks whether there exists a global state $\rho_{\rm S}$ in ${\rm S}$ compatible with all of them, that is, 
\begin{equation}
\label{eq:SMP}
\exists\; \rho_{\rm S}\geq 0 \;\text{   such that   }\;{\rm tr}_{\rm S\setminus X}(\rho_{\rm S}) = \rho_{\rm X},\, \forall\; {\rm X}. 
\end{equation}
\end{definition}
SMPs always have a trivial solution if the marginals do not overlap, $\rho_{\rm S}=\bigotimes_{\rm X\in\Lambda} \rho_{\rm X}$. The problem becomes much more interesting in the case of overlapping regions. It is again easy to see that the problem has no solution when states $\rho_{\rm X}$ are picked at random. This is because, for the problem to be well-posed, the states $\rho_{\rm X}$ must be compatible in the overlapping regions, that is, their reduced state must be equal. Formally, ${\rm tr}_{\rm X\setminus(X\cap Y)}(\rho_{\rm X}) = {\rm tr}_{\rm Y\setminus(X\cap Y)}(\rho_{\rm Y})\;\forall\;{\rm X,Y\in\Lambda}$. These are sometimes called local compatibility conditions, which are necessary and easy to verify but, unfortunately, not sufficient. Still, Eq.~\eqref{eq:SMP} is nothing but a set of linear equations on positive operators that can be solved using semi-definite programming (SDP), a standard technique in convex optimization. Note however that the size of the SDP scales exponentially with the number of systems.  In fact, the SMP is not expected to have a scalable solution as it is known to be QMA-complete~\cite{QMA}.

{\em Channel Marginal Problems.}--
To formalize the dynamical version of SMPs, first recall that channels correspond to the most general linear operation on quantum states and are represented by completely positive trace-preserving (CPTP) maps~\cite{QCIbook}, denoted by $\mE$. 
One is then tempted to define the CMP as follows:  given a set of local channels $\{\mE_{\rm X}\}_{\rm X\in\Lambda}$, where $\Lambda$ is a collection of subsystems ${\rm X}$ of ${\rm S}$ and each $\mE_{\rm X}$ is a channel in ${\rm X}$, the CMP asks whether there exists a global channel $\mE_{\rm S}$ compatible with  all of them. The formalization of this compatibility condition, however, is not obvious because the concept of {\em marginals} for channels is not as straightforward as that for states.

In fact, being an input-output process, the existence of well-defined marginals for a dynamical map requires certain no-signaling conditions. 
Consider the classical case first. 
A dynamical map from the the inputs ${\rm S'}=\{s'_i\}$ to the outputs ${\rm S}=\{s_j\}$ is defined by a stochastic matrix $P_{\rm S|S'}$~\cite{footnote:ClassicalChannelAsQuantumChannel}.
Given ${\rm X' \subset S'}$ and ${\rm X \subset S}$, $P_{\rm S|S'}$ has a well-defined marginal from ${\rm X'}$ to ${\rm X}$, or simply a reduced map $P_{\rm X|X'}$, if and only if for every input probability distribution $p_{\rm S'}$ we have
\begin{equation}\label{eq:classicalmarginal}
P_{\rm X|X'} \sum_{s'_i \in {\rm S'} \backslash {\rm X'}} p_{\rm S'} =  \sum_{s_i \in {\rm S} \setminus {\rm X}} P_{\rm S|S'}p_{\rm S'},
\end{equation}
where $P_{\rm B|A} p_{\rm A} = \sum_{a_1, ..., a_k} P({b_1, ..., b_k| a_1, ..., a_k})p_{a_1, ..., a_k}$. 
This is equivalent to the usual condition that the map $P_{\rm S|S'}$ is no-signaling from ${\rm S'} \setminus {\rm X'}$ to ${\rm X}$; namely, $\sum_{s_i \in {\rm S} \setminus {\rm X}} P_{\rm S|S'}$ is independent of the input ${\rm S'\setminus X'}$. 
\toni{For instance, in the standard bipartite scenario, where local inputs by Alice and Bob are now labelled by ${\rm X}$ and ${\rm Y}$ and their outputs by ${\rm A}$ and ${\rm B}$, a global map $P_{\rm AB|XY}$ has well-defined marginals $P_{\rm A|X}$ and $P_{\rm B|Y}$ if it satisfies the standard no-signaling conditions.}
Moving now to the quantum domain, the condition in Eq.~\eqref{eq:classicalmarginal} admits a straightforward quantum generalization. 
Denoting by $\mathcal{E}_{\rm A|A'}$ a channel from ${\rm A'}$ to ${\rm A}$ and using ${\rm A|A'}$ to represent this particular output-input pair, we have:

\begin{figure}
\scalebox{0.8}{\includegraphics{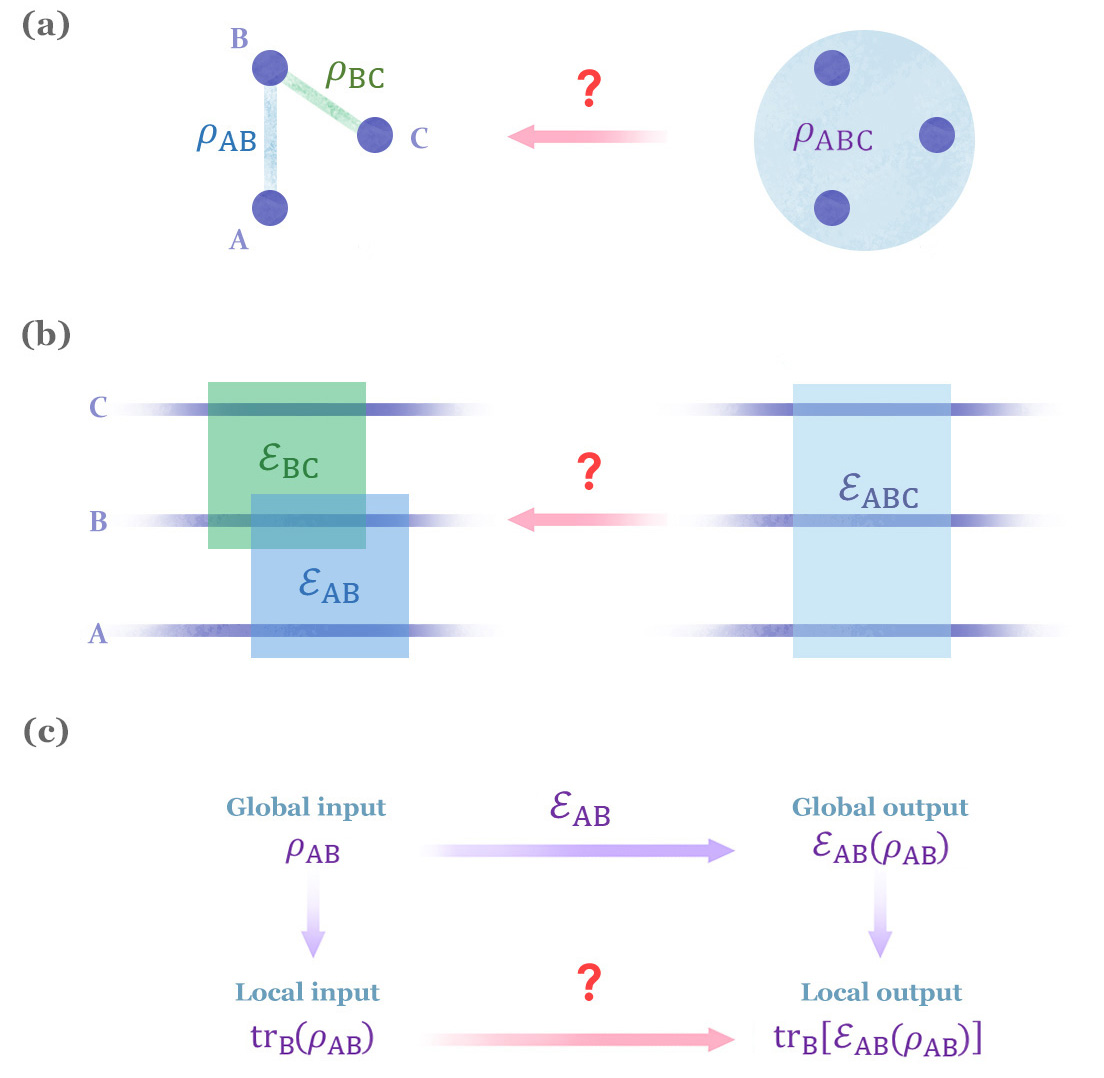}}
\caption{
(a) A {\em state marginal problem} is asking whether, for a given set of local states (e.g., $\rho_{\rm AB},\rho_{\rm BC}$), there is a global state ($\rho_{\rm ABC}$) that has them as its marginal states. 
(b) As a dynamical generalization of the state marginal problem, a {\em channel marginal problem} is asking whether, for a given set of local channels (e.g., $\mE_{\rm AB},\mE_{\rm BC}$, \CY{where $\mE_{\rm X} = \mE_{\rm X|X'}$ when ${\rm X=X'}$}), there is a global channel ($\mE_{\rm ABC}$) that has them as its marginal channels. 
(c) Commutation diagram used to define the {\em marginal channel} in a subsystem \CY{(Definition~\ref{Def:Reduced})}.
For a local input coming from part of a global input, ${\rm tr}_{\rm B}(\rho_{\rm AB})$, the marginal channel $\CY{{\rm Tr}_{\rm B|B}}(\mE_{\rm AB})$ \CY{of a global channel $\mE_{\rm AB}$} should map it to the marginal state of the global output, ${\rm tr}_{\rm B}[\mE_{\rm AB}(\rho_{\rm AB})]$.
Unlike the case for states, the marginal of a channel is not always well-defined.
}
\label{Fig} 
\end{figure}

\begin{definition}\label{Def:Reduced}
{\rm (Marginal Channels)} Given a global channel \CY{$\mE_{\rm S|S'}$} and \CY{subsystems ${\rm X'}\subseteq{\rm S'}, {\rm X}\subseteq{\rm S}$}, \CY{$\mE_{\rm S|S'}$} is said to have a well-defined marginal \CY{from ${\rm X'}$ to ${\rm X}$}, or reduced channel \CY{$\mE_{{\rm X|X'}}$}, if for every state $\rho_{\rm S'}$
\CY{\begin{align}\label{Eq:Def}
\mE_{\rm X|X'}\circ\tr_{{\rm S'}\setminus{\rm X'}}(\rho_{\rm S'}) = \tr_{{\rm S}\setminus{\rm X}}\circ\mE_{\rm S|S'}(\rho_{\rm S'}).
\end{align}
If this is the case, we denote the marginal channel by $\Tr_{{\rm S}\setminus{\rm X}|{\rm S'}\setminus{\rm X'}}(\mE_{\rm S|S'})\coloneqq\mE_{\rm X|X'}$.}
\end{definition}
\CY{
Channels satisfying this condition are known in the literature as \emph{semi-causal channels} in ${\rm X}$ with respect to ${\rm S'}\setminus{\rm X'}$ or {\em no-signaling} from ${\rm S'} \setminus {\rm X'}$ to ${\rm X}$~\cite{Beckman2001,Eggeling2002,Piani2006,Duan2016}.
Comparing Eq.~\eqref{Eq:Def} with \eqref{eq:classicalmarginal}, these channels can be seen as the analogue of no-signaling correlations once we quantize inputs and outputs (in fact, after dephasing inputs and outputs we recover the definition of no-signaling correlations). 
Operationally, these channels were also proven to be equivalent to those quantum dynamics that can be realized with one-way communication from ${\rm X'}$ to ${\rm S \setminus X}$~\cite{Eggeling2002} (also known as {\em semi-localizable} in ${\rm X}$). 
Unsurprisingly, this class of channels plays a critical role in the study of the quantum generalization of causal models~\cite{Allen2017}. What we highlight here is that they are also at the core of the dynamical generalization of the SMP, which can be written in full analogy to Definition~\ref{Def:SMP}:}
\begin{definition}\label{Def:CMP}
{\rm (Channel Marginal Problem)} Consider global systems \CY{${\rm S', S}$} and a set of local channels \CY{$\{\mE_{\rm X|X'}\}_{\rm X|X'\in\Lambda}$}, \CY{where $\Lambda\coloneqq\{{\rm X|X'}\}$ is a collection of \CYtwo{output-input pairs} with ${\rm X'\subseteq S', X\subseteq S}$}.
Then a CMP asks whether there exists a global channel \CY{$\mE_{\rm S|S'}$} compatible with all of them, 
\begin{equation}
\label{Eq:CMP}
\exists \text{ channel }\CY{\mE_{\rm S|S'}} \;\text{s.t.}\;\CY{\Tr_{\rm S\setminus X|S'\setminus X'}(\mE_{\rm S|S'}) = \mE_{\rm X|X'}}\, \forall\; \CY{\rm X|X'}. 
\end{equation}
\end{definition}

With this definition, the analogies between the CMP and SMP are clear. 
As for states, when the subsets ${\rm X}$'s \CY{and ${\rm X'}$}'s are non-overlapping, the CMP has a trivial solution, namely, $\CY{\bigotimes_{\rm X|X'\in\Lambda}\mE_{{\rm X|X'}}}$. 
When overlapping marginals are considered, we again need to first verify whether the problem is well-posed and the overlapping channels coincide in the common region: 
\CY{ For every ${\rm X|X',Y|Y'}\in\Lambda$
we need,}
$
\CY{\Tr_{\rm X\setminus Y|X'\setminus Y'}(\mE_{\rm X|X'}) = \Tr_{\rm Y\setminus X|Y'\setminus X'}(\mE_{\rm Y|Y'}).}
$
A set of channels \CY{$\{\mE_{{\rm X|X'}}\}_{\rm X|X'\in\Lambda}$ satisfying this condition is said to be {\em locally compatible}.}{
As for states, we will see below that these conditions are necessary but not sufficient.

Our next step is to see how condition~\eqref{Eq:CMP} can be tackled and, in particular, whether it can be written as a SDP. 
But before doing that, we \CY{showcase how the CMP includes as special cases several problems considered before in the quantum information and physics literature.}

\CY{{\em Special cases.--} First consider the case where all ${\rm X'}$ coincide with ${\rm S'}$; namely, $\Lambda = \{{\rm X|S'}\}$. 
The global channels $\mathcal{E}_{\rm S|S'}$ automatically have well-defined marginals $\mathcal{E}_{\rm X|S'}$, since the no-signaling condition in Definition~\ref{Def:Reduced} trivializes. 
The CMP then reduces to the question of the existence of a global channel $\mathcal{E}_{\rm S|S'}$ such that ${\rm tr}_{\rm S\setminus X}\circ\mE_{\rm S|S'} = \mE_{\rm X|S'}\;\forall\;{\rm X}$. 
This broadcasting notion of channel (in)compatibility has been studied extensively as a natural generalization of the notion of measurement (in)compatibility \cite{Haapasalo2014,Girard2020,Viola2015, Haapasalo2019,Haapasalo2015,Carmeli2019,Carmeli2019-2}.
When the channels under consideration are the identity, the non-existence of a global channel is the celebrated no-broadcasting theorem~\cite{Heinosaari2016}.

Next, consider the case where ${\rm X'}_i= {\rm X}_i = {\rm A B}_i$, where all ${\rm B}_i$'s are isomorphic. 
Given identical channels $\mathcal{E}_{{\rm AB}_i}$, the CMP asks whether there exists an extension to a global channel $\mathcal{E}_{{\rm A B}_1 \dots {\rm B}_n}$. 
This notion of channel extendibility was recently introduced as a natural generalization of the notion of state extendibility and used in the context of quantum communication scenarios~\cite{Kaur2019} and testing symmetries on a quantum computer~\cite{LaBorde2021}.}

{\em Solving the CMP.}-- In what follows, and for the ease of notation, we often denote the vector of local channels defining a given CMP by $\mathbfcal{E}\coloneqq\CY{\{\mE_{{\rm X|X'}}\}_{\rm X|X'\in\Lambda}}$.
We say that $\mathbfcal{E}$ is {\em compatible} whenever the CMP has a solution, namely there exists a global channel $\CY{\mE_{\rm S|S'}}$ compatible with each $\CY{\mE_{{\rm X|X'}}}$. 
The set of compatible local channels is denoted by $\mathfrak{C}$.
As a first step to solve the CMPs, we consider the Choi-Jamio\l kowski isomorphism between channels and bipartite states. The {\em Choi state} of a channel $\CY{\mE_{\rm X|X'}}$ is defined by~\cite{Choi1975,Jamiolkowski1972}
\begin{align}
\mE_{\rm XX'}^\mathcal{J}\coloneqq \CY{(\mE_{\rm X|X'}\otimes\mathcal{I}_{\rm X'})(\proj{\Psi_{\rm X'X'}^+}),}
\end{align}
\CY{where $\ket{\Psi_{\rm X'X'}^+}\coloneqq\frac{1}{\sqrt{d_{\rm X'}}}\sum_{i=0}^{d_{\rm X'}-1}\ket{ii}$ is maximally entangled in ${\rm X'X'}$, which is a bipartite system consisting of two copies of ${\rm X'}$.}
Note that, since channels are trace-preserving, $\tr_{\rm X}(\mE_{\rm XX'}^\mathcal{J})=\frac{\id_{\rm X'}}{d_{\rm X'}}$.
It has been shown~\cite{Beckman2001,Eggeling2002,Piani2006,Duan2016,Wang2017,Hoban2018} that the condition~\eqref{Eq:Def} used to define a marginal channel \CY{$\mE_{\rm X|X'}$} of a global channel \CY{$\mE_{\rm S|S'}$} can equivalently be expressed in terms of their Choi states as
\begin{align}\label{Eq:Characterization}
\tr_{\rm S\setminus X}\left(\mE_{\rm SS'}^\mathcal{J}\right) = \mE_{\rm XX'}^\mathcal{J}\otimes\frac{\id_{\rm S'\setminus X'}}{d_{\rm S'\setminus X'}},
\end{align}
which we prove for completeness in Appendix~\ref{App:cmp}.

It is now straightforward to present an equivalent formulation of the CMP in term of the Choi states: a set of local channels $\mathbfcal{E}$ is compatible if, and only if,
\begin{equation}\label{Eq:CMPSDP}
\exists\;\mE_{\rm SS'}^\mathcal{J}\geq 0\;\;
{\rm s.t.}\;\;\tr_{\rm S\setminus X}\left(\mE_{\rm SS'}^\mathcal{J}\right) = \mE_{\rm XX'}^\mathcal{J}\otimes\frac{\id_{\rm S'\setminus X'}}{d_{\rm S'\setminus X'}}\;\forall\;\CY{\rm X|X'}.
\end{equation}
Note that there is no need to impose that $\mE_{\rm SS'}^\mathcal{J}$ is a Choi state, as the condition $\tr_{\rm S}(\mE_{\rm SS'}^\mathcal{J})=\frac{\id_{\rm S'}}{d_{\rm S'}}$ follows from the equalities in the previous formulation. 
Hence, the CMP can be rephrased via the Choi-Jamio\l kowski isomorphism as a particular SMP with overlapping marginals, but it requires an additional tensor product structure \CY{taking into account the quantum non-signaling constraints associated to the dynamical problem}. 
\CY{Eq.~\eqref{Eq:CMPSDP} involves a set of linear conditions over positive operators and can thus be solved using SDP, but it is not equivalent to the SMP for the Choi states.
However in the special case of broadcasting compatibility, where all ${\rm X'}$ coincide with the input global system ${\rm S'}$, the `tensor identity' parts in Eq.~\eqref{Eq:CMPSDP} disappear. 
In this special case the CMP reduces to a SMP for the Choi states, recovering the result of Ref.~\cite{Haapasalo2019}.}

\CY{{\em Channel Incompatibility Robustness.--}}
In last years, there has been an intense effort to provide operational measures to problems defined through SDP or conic programs, see for example~\cite{Hall2007,Rosset2018,Carmeli2019-2,Carmeli2019,Uola2020,Takagi2019-1,Uola2019,Skrzypczyk2019,Mori2020,Yuan2019}. 
The CMP defined in Eq.~\eqref{Eq:CMPSDP} is another example of these problems and we can therefore apply to it techniques similar to those presented in the previous references. We first present a robustness quantity, dubbed {\em incompatibility robustness}, which gives an efficient solution to the CMPs while also providing a quantitative measure of incompatibility:
\begin{align}\label{Eq:RI}
R(\mathbfcal{E})\coloneqq \max\{\lambda\in[0,1]\,|\,\lambda\mathbfcal{E} + (1-\lambda)\mathbfcal{N}\in\mathfrak{C}\},
\end{align}
where the maximization is taken over vectors of local channels \CY{$\mathbfcal{N} = \{\mathcal{N}_{\rm X|X'}\}_{\rm X|X'\in\Lambda}$.}
The linear combination of $\mathbfcal{E},\mathbfcal{N}$ is defined component-wise, that is as \CY{$a\mathbfcal{E} + b\mathbfcal{N}\coloneqq\{a\mE_{\rm X|X'} + b\mathcal{N}_{\rm X|X'}\}_{\rm X|X'\in\Lambda}$.}
\CY{$R(\mathbfcal{E})$ is the solution of the SDP (see Appendix~\ref{App:DualWitness})}
\begin{eqnarray}\label{Eq:Result:Conjecture2}
\begin{aligned}
\max_{\rho_{\rm SS'},\lambda}\quad&\lambda\\
{\rm s.t.}\quad&\rho_{\rm SS'}\ge0,\;\;\tr_{{\rm S}}(\rho_{\rm SS'}) = \frac{\id_{\rm S'}}{d_{\rm S'}},\;\;\lambda\in[0,1],\;\&\;\forall\;{\rm \CY{X|X'}}\\
&\tr_{{\rm S}\setminus{\rm X}}(\rho_{\rm SS'})\ge\lambda \mE_{{\rm X}{\rm X}'}^\mathcal{J}\otimes\frac{\id_{{\rm S}'\setminus{\rm X}'}}{d_{{\rm S}'\setminus{\rm X}'}};\\
&\tr_{{\rm S}\setminus{\rm X}}(\rho_{\rm SS'}) = \tr_{{\rm SS'}\setminus{\rm XX'}}(\rho_{\rm SS'})\otimes\frac{\id_{{\rm S}'\setminus{\rm X}'}}{d_{{\rm S}'\setminus{\rm X}'}}.
\end{aligned}
\end{eqnarray}
Eq.~\eqref{Eq:Result:Conjecture2} provides a general, quantitative and numerically feasible strategy to tackle CMPs.
Its solution $\lambda$ serves as a measure of incompatibility: $R(\mathbfcal{E})=1$ if and only if $\mathbfcal{E}\in\mathfrak{C}$, meaning that the CMP for $\mathbfcal{E}$ admits a solution. 
A value $R<1$ can detect instances in which the local compatibility of channels \CY{$\mE_{\rm X|X'}$} is not sufficient for the existence of the global channel \CY{$\mE_{\rm S|S'}$}~\cite{footnote:LC}. 
The solution to the SDP also returns the global physical process that best approximates the marginal channels~$\mathbfcal{E}$.
 \CY{In the broadcasting scenario (${\rm X'=S'}$), the incompatibility robustness recovers as a special case the `consistent robustness' of \cite{Haapasalo2019} (and the results of Theorem~1 of the same manuscript).}

The dual of the SDP~\eqref{Eq:Result:Conjecture2} provides a simple operational interpretation, helping us to single out the physical role played by incompatibility.
To proceed, let $\mathbfcal{E}^\mathcal{J}\coloneqq\{\mE^\mathcal{J}_{\rm XX'}\}_{\rm \CY{X|X'\in\Lambda}}$ be the vector of Choi states of the input channels $\mathbfcal{E}$. For a set of operators \CY{${\bf A}\coloneqq\{A_{\rm X|X'}\}_{\rm X|X'\in\Lambda}$in ${\rm XX'}$}, define $\left\langle{\bf A},\mathbfcal{E}^\mathcal{J}\right\rangle\coloneqq\sum_{\rm \CY{\rm X|X'\in\Lambda}}{\rm tr}\left(\CY{A_{\rm X|X'}}\mE^\mathcal{J}_{\rm XX'}\right)$.
Then we prove that $\mathbfcal{E}$ is incompatible if and only if there exists a set of positive operators \CY{${\bf H}\coloneqq\{H_{\rm X|X'}\}_{\rm X|X'\in\Lambda}$ \CY{in ${\rm XX'}$}} such that (see Appendix~\ref{App:DualWitness})
\begin{align}\label{Eq:Result:CompatibilityProblemCondition}
\left\langle{\bf H},\mathbfcal{E}^\mathcal{J}\right\rangle > \max_{\mathbfcal{L}\in\mathfrak{C}}\left\langle{\bf H},\mathbfcal{L}^\mathcal{J}\right\rangle.
\end{align}
This result gives a witness form for channel incompatibility in terms of Choi states of local channels, serving as the dynamical version of the state incompatibility witness given by Ref.~\cite{Hall2007}.
However, intuitively one expects that it is possible to define witnesses in terms of channels, rather than Choi states.
Using Proposition 7 in Ref.~\cite{Rosset2018}, in Appendix~\ref{App:ProofsMainResults} we prove the following characterization of incompatibility:
\begin{theorem}\label{Coro}
{\rm (Channel Incompatibility Witness)} $\mathbfcal{E}$ is incompatible if and only if for every ${\rm X|X'\in\Lambda}$ there exist 
Hermitian operators $\{H^{(i)}_{{\rm X|X'}}\}_{i=1}^{N_{\Lambda}}$ in ${\rm X}$ and states $\{\rho^{(i)}_{{\rm X|X'}}\}_{i=1}^{N_{\Lambda}}$ in ${\rm X'}$ such that, with $N_{\Lambda} \coloneqq \left(\max_{\rm X|X'\in\Lambda}\{d_{\rm X},d_{\rm X'}\}\right)^2 +3$,
\begin{align}
&\sum_{{\rm X|X'},i}{\rm tr}\left[H^{(i)}_{{\rm X|X'}}\mE_{\rm X|X'}\left(\rho^{(i)}_{{\rm X|X'}}\right)\right]\nonumber\\
& > \max_{\mathbfcal{L}\in\mathfrak{C}}\sum_{{{\rm X|X'},i}}{\rm tr}\left[H^{(i)}_{{\rm X|X'}}\mathcal{L}_{\rm X|X'}\left(\rho^{(i)}_{{\rm X|X'}}\right)\right].
\end{align}
\end{theorem}

{\em Examples of incompatibility.}--
In analogy to states, there exist locally compatible channels for which there exists no global evolution compatible with them.
Consider channels ({\rm X = A,C})
$
\mathcal{M}_{\rm XB}(\cdot)\coloneqq{\rm CNOT}_{\rm XB}\left[\proj{0}_{\rm X}\otimes{\rm tr}_{\rm X}(\cdot)\right],
$
where ${\rm CNOT}_{\rm XB}:\ket{i}_{\rm X}\otimes\ket{j}_{\rm B}\mapsto\ket{(i+j)_{\rm mod\;2}}_{\rm X}\otimes\ket{j}_{\rm B}$ is a qubit CNOT-gate. 
$\mathcal{M}_{\rm AB}$ can be realized by Bob (i) implementing a CNOT between his incoming particle and an extra ancillary qubit prepared in state $\ket{0}$ (controlled on the former) and (ii) sending the ancillary qubit to ${\rm A}$, who uses it as output. 
This is incompatible with $\mathcal{M}_{\rm CB}$, as the SDP~\eqref{Eq:Result:Conjecture2} returns $R(\{\mathcal{M}_{\rm AB},\mathcal{M}_{\rm CB}\}) = 0.75$. 
The SDP also returns the optimal noise making the channels compatible:
$
\mathcal{N}_{\rm XB}(\cdot) = \frac{2}{3}\left(\mathcal{I}_{\rm X}\otimes\mathcal{D}_{\rm B}\right)\circ{\rm CNOT}_{\rm XB}\left[\proj{1}_{\rm X}\otimes{\rm tr}_{\rm X}(\cdot)\right]+\frac{1}{3}{\rm CNOT}_{\rm XB}^-\left[\proj{0}_{\rm X}\otimes{\rm tr}_{\rm X}(\cdot)\right],
$
where ${\rm CNOT}_{\rm XB}^-:\ket{i}_{\rm X}\otimes\ket{j}_{\rm B}\mapsto(-1)^j\ket{(i+j)_{\rm mod\;2}}_{\rm X}\otimes\ket{j}_{\rm B}$ and $\mathcal{D}(\cdot)\coloneqq\sum_{i=0,1}\proj{i}\cdot\proj{i}$. 
In fact, this is the noise generated by a physical process involving the best universal cloning machine allowed by quantum mechanics~\cite{Cloning-RMP}, which corresponds to the following protocol: After performing the CNOT between ${\rm X}$ and B, Bob sends ${\rm X}$ into an optimal cloning machine. 
He then sends one (imperfect) clone to Alice and one to Charlie, who use them as outputs. 
See Appendix~\ref{App:Example} for details.

{\em CMP is irreducible to SMP.--} One may conjecture that the compatibility of channels $\mathbfcal{E} = \{ \mE_{\rm X|X'}\}$ can be reduced to the state compatibility of the image states $\{ \mE_{\rm X|X'}(\rho_{\rm X'})\}$ for every set of compatible inputs $\{\rho_{\rm X'}\}$. 
We disprove this conjecture by constructing a counterexample. 
Take the bipartite channel $\mathcal{K}_{\rm XB}(\cdot)\coloneqq{\rm SWAP}\left[\proj{0}_{\rm X}\otimes{\rm tr}_{\rm X}(\cdot)\right]$, where ${\rm SWAP}:\ket{ij}\mapsto\ket{ji}$ is the swap operation.
First, $\{\mathcal{K}_{\rm AB},\mathcal{K}_{\rm CB}\}$ is locally compatible since in ${\rm B}$ they are both $(\cdot)_{\rm B}\mapsto\proj{0}_{\rm B}$. 
This pair is nevertheless incompatible.
To see this, by contradiction suppose that there was a tripartite channel $\mathcal{K}_{\rm ABC}$ simultaneously realizing $\mathcal{K}_{\rm AB},\mathcal{K}_{\rm CB}$.
By considering all inputs of the form $\proj{\psi}_{\rm B}\otimes\proj{00}_{\rm AC}$, $\mathcal{K}_{\rm ABC}$ realizes $\ket{\psi}\mapsto\ket{\psi}\otimes\ket{\psi}$, in violation of the no-cloning theorem~\cite{Cloning-RMP,Wootters1982}.
Hence, $\mathcal{K}_{\rm ABC}$ does not exist.
Still, for every compatible input pair $\{\rho_{\rm AB},\rho_{\rm CB}\}$ with marginal in ${\rm B}$ as $\sigma = {\rm tr}_{\rm A}(\rho_{\rm AB}) = {\rm tr}_{\rm C}(\rho_{\rm BC})$, the image states under $\{\mathcal{K}_{\rm AB},\mathcal{K}_{\rm CB}\}$ read $\{\sigma_{\rm A}\otimes\proj{0}_{\rm B},\proj{0}_{\rm B}\otimes\sigma_{\rm C}\}$, which are compatible with $\sigma_{\rm A}\otimes\proj{0}_{\rm B}\otimes\sigma_{\rm C}$.
Hence, incompatibility of local channels cannot be always detected from incompatibility of their image states~\cite{footnote:EntMonogamy}.
See Appendix~\ref{App:CMP_Irreducible_to_SMP} for further remarks.

{\em Gap Between Classical and Quantum CMP.--} The previous example demonstrated the existence of non-trivial instances of quantum channel incompatibility in the AB/BC scenario. Remarkably, this incompatibility structure never occurs classically. 
In fact, given two classical channels  (i.e., stochastic matrices) $P_{\rm AB|XY}$, $P_{\rm BC|YZ}$ with a well-defined marginal $P_{\rm B|Y}$ in ${\rm B|Y}$, the CMP is always satisfied by the global classical channel
$
P_{\rm ABC|XYZ} = P_{\rm AB|XY}P_{\rm BC|YZ}/P_{\rm B|Y},
$
as one can show by taking the corresponding marginals~\cite{footnote:ClassicalChannels}.

Is the gap between classical and quantum CMP simply due to the fact that the latter is always trivial, that is, simply defined by local compatibility? 
The following example shows that this is not the case. 
Take a PR-box~\cite{Khalfi1985,Popescu1994} on AB: $P_{\rm AB|XY} = \frac{1}{2}$ if ${\rm A} \oplus {\rm B} = {\rm XY}$ and $0$ otherwise, where all random variables are bits. Similarly, define PR-boxes in ${\rm AC}$ and ${\rm BC}$. 
The marginals are well-defined and coincide on ${\rm A}$, ${\rm B}$ and ${\rm C}$. 
Now, is there a joint classical channel $P_{\rm ABC|XYZ}$ compatible with them? 
By contradiction, suppose it does. 
Note that all its $2$-party marginals are well-defined, so one can prove that $P_{\rm ABC|XYZ}$ must be a no-signaling distribution. 
But it is known that the PR-box cannot be shared without violating no-signaling condition~\cite{Barrett2005}, so $P_{\rm ABC|XYZ}$ cannot exist. 
This proves that the channel marginal problem is not trivial either, but it is structurally different from the quantum one. In the classical problem dynamical incompatibility is a consequence of loops, while quantumly these are not required. This mirrors at the dynamical level what happens with frustration in the classical SMP~\cite{Wolf2003}.

{\em Operational interpretations.}--
Finally, we can also relate the CMP to discrimination tasks. 
In fact, a direct application of the formalism of Ref.~\cite{Uola2020} allows interpreting the robustness $R(\mathbfcal{E})$ in terms of an input-output game. 
Here we are more interested in standard {\em ensemble state discrimination task}. 
Given $\Lambda$, consider a scenario which has an agent for each ${\rm X}$ and ${\rm X}'$. With  probability $p_{\rm X|X'}$ the output-input pair $\rm X|X'$ is announced and the agent in ${\rm X'}$ is given a set of states $\{\rho_{{\rm X|X'}}^{(i)}\}_i$ sent with probabilities $\{q_{{\rm X|X'}}^{(i)}\}_i$ (with $\sum_iq_{{\rm X|X'}}^{(i)}=1$). 
The agent in ${\rm X'}$ needs to send them to the agent in ${\rm X}$ via an arbitrary channel $\mE_{\rm X|X'}$. The agent in ${\rm X}$ receives the states and performs a discriminating measurement $M^{(i)}_{{\rm X|X'}}$ (formally, $M^{(i)}_{{\rm X|X'}}\ge0$ and $\sum_i M^{(i)}_{{\rm X|X'}}=\id_{\rm X}$), guessing state $i$ in case of outcome~$i$. 
Setting $D\coloneqq\left(\{p_{\rm X|X'}\},\{q^{(i)}_{{\rm X|X'}},\rho^{(i)}_{{\rm X|X'}}\},\{M^{(i)}_{{\rm X|X'}}\}\right)$, the corresponding success probability in the task reads
$
P(D,\mathbfcal{E})\coloneqq\sum_{\rm X|X'}p_{\rm X|X'}\sum_{i} q^{(i)}_{{\rm X|X'}}{\rm tr}\left[M^{(i)}_{{\rm X|X'}}\mE_{\rm X|X'}\left(\rho^{(i)}_{{\rm X|X'}}\right)\right].
$
To avoid trivial scenarios, take $D$ to be {\em strictly positive}, meaning \CY{$p_{\rm X|X'}>0, q^{(i)}_{{\rm X|X'}}>0, M^{(i)}_{{\rm X|X'}}>0\;\forall \,i\;\forall\,{\rm X|X'}$} \cite{footnote:strictly}.
Let $P_\mathfrak{C}(D)\coloneqq\max_{\mathbfcal{E}\in\mathfrak{C}}P(D,\mathbfcal{E})$ be the highest success probability achievable by compatible sets of channels in the discrimination task $D$~\cite{footnote:NSD}. 
Then incompatibility of channels is equivalent to an advantage in an ensemble state discrimination task (proof in Appendix~\ref{App:ProofsMainResults}):
\begin{theorem}\label{Result:EnsembleStateDiscrimination}
{\rm (Advantage in Discrimination Tasks)} $\mathbfcal{E}$ is incompatible if and only if there exists a strictly positive ensemble state discrimination task $D$ such that
$
P(D,\mathbfcal{E}) > P_\mathfrak{C}(D).
$
\end{theorem}
{\em Implications of Theorem~\ref{Result:EnsembleStateDiscrimination}.--}
Setting ${\rm X'=S'}$ for every member in $\Lambda$, Theorem~\ref{Result:EnsembleStateDiscrimination} implies that every set of broadcast incompatible channels gives an advantage over compatible ones in some ensemble state discrimination tasks.
This recovers results from Refs.~\cite{Carmeli2019,Uola2019}.
Also, in the particular case of quantum-to-classical channels, Theorem~\ref{Result:EnsembleStateDiscrimination} recovers results on the discrimination advantages of incompatible measurements of Refs.~\cite{Carmeli2019-2,Mori2020,Skrzypczyk2019}.
}

{\em Conclusions.}--
We introduced the dynamical generalization of state marginal problems, termed {\em quantum channel marginal problems} (CMP), which is formulated as a marginal problem for channels. 
The CMP involves in its definition the quantum generalization of the no-signaling conditions. 
It encompasses and naturally generalizes several notions of quantum incompatibility \cite{Haapasalo2014,Girard2020, Haapasalo2019,Haapasalo2015,Carmeli2019,Carmeli2019-2} and channel extendibility \cite{Kaur2019, LaBorde2021}.

The problem can be expressed in terms of states using the Choi-Jamio\l kowski isomorphism, but in general it is irreducible to state marginal problems.
We provided a necessary and sufficient condition for channel incompatibility in terms of a robustness measure, which can be cast into a semidefinite program.
A witness form for dynamical incompatibility can be derived, giving channel incompatibility an operational interpretation in a state discrimination task which recovers several recent results as special cases \cite{Carmeli2019,Uola2019, Carmeli2019-2,Mori2020,Skrzypczyk2019}. 
The study of CMPs not only describes within a unified framework previously disconnected notions of incompatibility that found applications in quantum communication~\cite{Kaur2019}, quantum foundations~\cite{Heinosaari2016} and computing~\cite{LaBorde2021}. 
In analogy to the study of no-signaling correlations, it may also bring novel insights into the dynamical structure of quantum theory as it differs from both classical and supra-quantum theories~\cite{Hoban2018}. 
In this context, we discussed a gap between classical and quantum CMP.

{\em Acknowledgements.}--
We thank (in alphabetical order) Paolo Abiuso, Teiko Heinosaari, Felix Huber, Yeong-Cherng Liang, Kavan Modi, Stefano Pironio, Jacopo Surace, M\'arcio Taddei, and Roope Uola for fruitful discussions and comments.
This project is part of the ICFOstepstone - PhD Programme for Early-Stage Researchers in Photonics, funded by the Marie Sk\l odowska-Curie Co-funding of regional, national and international programmes (GA665884) of the European Commission.
We also acknowledge support from the ERC AdG CERQUTE, the AXA Chair in Quantum Information Science, the Government of Spain (FIS2020-TRANQI and Severo Ochoa CEX2019-000910-S), Fundacio Cellex, Fundacio Mir-Puig, Generalitat de Catalunya (CERCA, AGAUR SGR 1381), EU Marie Sk\l odowska-Curie (H2020-MSCA-IF-2017, GA794842) and grant EQEC No. 682726.

\appendix\onecolumngrid

\section{Proofs for the Section ``Solving the CMP''}\label{App:cmp}
In this section, we prove Eq.~\eqref{Eq:Characterization} in the main text, which can be formally stated as the following lemma:
\begin{alemma}\label{Lemma:Characterization}
$\CY{\mE_{\rm S|S'}}$ is compatible with $\CY{\mE_{\rm X|X'}}$ if and only if 
\CYtwo{
$\tr_{\rm S\setminus X}\left(\mE_{\rm SS'}^\mathcal{J}\right) = \mE_{\rm XX'}^\mathcal{J}\otimes\frac{\id_{\rm S'\setminus X'}}{d_{\rm S'\setminus X'}}.$
}
\end{alemma}
This result has appeared in several previous works~\cite{Beckman2001,Eggeling2002,Piani2006,Duan2016,Wang2017,Hoban2018}, here we provide a proof for the sake of completeness.
\begin{proof}
For every \CY{input state $\rho_{\rm X'}$} \CYtwo{in ${\rm X'}$}, one can write $\CY{\mE_{\rm X|X'}(\rho_{\rm X'}) = d_{\rm X'}}{\rm tr}_{\rm X'}[(\id_{\rm X}\otimes\rho_{\rm X'}^T)\CYtwo{\mE_{\rm XX'}^\mathcal{J}}],$ where $(\cdot)^T$ is the transpose map.
This means
\begin{align}
\tr_{{\rm S}\setminus{\rm X}}\circ\CY{\mE_{\rm S|S'}(\rho_{\rm S'})}= \CY{d_{\rm S'}}\tr_{{\rm S}\setminus{\rm X}}\circ\tr_{{\rm S'}}\left[\left(\id_{\rm S}\otimes\rho_{{\rm S'}}^T\right)\mE_{\rm SS'}^\mathcal{J}\right]=\CY{d_{\rm S'}}\tr_{\rm S'}\left[\left(\id_{{\rm X}}\otimes\rho_{\rm S'}^T\right)\tr_{{\rm S}\setminus{\rm X}}\left(\mE^\mathcal{J}_{\rm SS'}\right)\right],
\end{align}
On the other hand, we have
\begin{align}
\CY{\mE_{{\rm X|X'}}\circ\tr_{{\rm S'}\setminus{\rm X'}}(\rho_{\rm S'})} = \CY{d_{{\rm X'}}}\tr_{{\rm X}'}\left\{\left[\id_{{\rm X}}\otimes\left(\tr_{{\rm S'}\setminus{\rm X}'}(\rho_{\rm S'})\right)^T\right]\mE^\mathcal{J}_{{\rm X}{\rm X}'}\right\}
 = \CY{d_{{\rm X'}}}\tr_{\rm S'}\left[\left(\id_{{\rm X}}\otimes\rho_{\rm S'}^T\right)\left(\mE^\mathcal{J}_{{\rm X}{\rm X}'}\otimes\id_{{\rm S'}\setminus{\rm X}'}\right)\right],
\end{align}
where \CYtwo{we use} the identity ${\rm tr}_{\rm A}(\rho_{\rm AB})^T = {\rm tr}_{\rm A}(\rho_{\rm AB}^T)$~\cite{footnote1}.
Since $\CY{d_{{\rm X'}}}d_{{\rm S}'\setminus{\rm X}'} = d_{\rm S'}$, the validity of Eq.~\eqref{Eq:Characterization} in the main text implies that \CYtwo{$\mE_{\rm S|S'}$ is compatible with $\mE_{{\rm X|X'}}$}.
This proves that Eq.~\eqref{Eq:Characterization} is sufficient for the compatibility.

To show that it is also necessary, we note that the validity of \CY{$\mE_{\rm X|X'}\circ\tr_{{\rm S'}\setminus{\rm X'}}(\rho_{\rm S'}) = \tr_{{\rm S}\setminus{\rm X}}\circ\mE_{\rm S|S'}(\rho_{\rm S'})$} implies
\begin{align}
 \tr_{{\rm S}\setminus{\rm X}}\left(\mE_{\rm SS'}^\mathcal{J}\right)= \left[\left(\tr_{{\rm S}\setminus{\rm X}}\circ\CY{\mE_{\rm S|S'}}\right)\otimes\mathcal{I}_{\rm S'}\right](\CY{\Psi^+_{\rm S'S'}})
=\left[\left(\CY{\mE_{\rm X|X'}\circ\tr_{{\rm S'}\setminus{\rm X'}}}\right)\otimes\mathcal{I}_{\rm S'}\right](\CY{\Psi^+_{\rm S'S'}})
\CYtwo{=\mE^\mathcal{J}_{\rm XX'}\otimes\frac{\id_{\rm S'\setminus X'}}{d_{\rm S'\setminus X'}}.}
\end{align}
This completes the proof.
\end{proof}

Applying Lemma~\ref{Lemma:Characterization} to ${\rm X,Y}$ and ${\rm X\cap Y}$ \CY{(and hence also ${\rm X',Y'}$ and ${\rm X'\cap Y'}$)} shows a characterization for local compatibility.
\CYtwo{
More precisely, we have the following corollary:
\begin{acorollary}
If $\mathcal{E}_{\rm X|X'}$ and $\mathcal{E}_{\rm Y|Y'}$ are compatible, then
$
{\rm Tr}_{\rm X\setminus Y|X'\setminus Y'}(\mathcal{E}_{\rm X|X'}) = {\rm Tr}_{\rm Y\setminus X|Y'\setminus X'}(\mathcal{E}_{\rm Y|Y'}) ;
$
namely, they are locally compatible.
\end{acorollary}
\begin{proof}
To see this, we redefine the systems as ${\rm A=X\setminus Y},{\rm B =X\cap Y},{\rm C=Y\setminus X},{\rm S = A\cup B\cup C}$ (similar definitions apply to ${\rm A',B',C',S'}$).
Suppose that $\mathcal{E}_{\rm X|X'} = \mathcal{E}_{\rm AB|A'B'},\mathcal{E}_{\rm Y|Y'} = \mathcal{E}_{\rm BC|B'C'}$ are compatible, then there exists a global channel $\mE_{\rm S|S'}$ such that
${\rm Tr}_{\rm C|C'}(\mathcal{E}_{\rm S|S'}) = \mathcal{E}_{\rm AB|A'B'}$ and ${\rm Tr}_{\rm A|A'}(\mathcal{E}_{\rm S|S'}) = \mathcal{E}_{\rm BC|B'C'}.$
Then Lemma~\ref{Lemma:Characterization} implies that
$
{\rm tr}_{\rm A}\left(\mathcal{E}_{\rm XX'}^\mathcal{J}\right)\otimes\frac{\id_{\rm C'}}{d_{\rm C'}} = {\rm tr}_{\rm AC}\left(\mathcal{E}_{\rm SS'}^\mathcal{J}\right)={\rm tr}_{\rm C}\left(\mathcal{E}_{\rm YY'}^\mathcal{J}\right)\otimes\frac{\id_{\rm A'}}{d_{\rm A'}}.
$
\CYthree{Tracing out ${\rm C'}$ and using Lemma~\ref{Lemma:Characterization} again, we learn that $\mathcal{E}_{\rm X|X'}$ is compatible with a marginal channel in ${\rm B|B'}$ [${\rm Tr}_{\rm X\setminus Y|X'\setminus Y'}(\mathcal{E}_{\rm X|X'})$] with Choi state ${\rm tr}_{\rm CC'}\left(\mathcal{E}_{\rm YY'}^\mathcal{J}\right) =  {\rm tr}_{\rm ACA'C'}\left(\mathcal{E}_{\rm SS'}^\mathcal{J}\right)$.
A similar argument (tracing out ${\rm A'}$) shows that $\mathcal{E}_{\rm Y|Y'}$ is compatible with a marginal channel in ${\rm B|B'}$ [${\rm Tr}_{\rm Y\setminus X|Y'\setminus X'}(\mathcal{E}_{\rm Y|Y'})$] with the same Choi state ${\rm tr}_{\rm AA'}\left(\mathcal{E}_{\rm XX'}^\mathcal{J}\right) = {\rm tr}_{\rm ACA'C'}\left(\mathcal{E}_{\rm SS'}^\mathcal{J}\right)$.
From here we conclude that ${\rm Tr}_{\rm X\setminus Y|X'\setminus Y'}(\mathcal{E}_{\rm X|X'}) = {\rm Tr}_{\rm Y\setminus X|Y'\setminus X'}(\mathcal{E}_{\rm Y|Y'})$.
}
\end{proof}
}

\section{Proofs for the Section ``Channel Incompatibility Robustness''}\label{App:DualWitness}

\subsection{Proof of Eq.~(\ref{Eq:Result:Conjecture2})}
\begin{proof}
By definition, $R(\mathbfcal{E})$ is the solution of 
\begin{eqnarray}\label{Eq:AppMaxForm}
\begin{aligned}
\max_{\mathbfcal{N},\mathbfcal{L},\lambda}\quad&\lambda\\
{\rm s.t.}\quad&\mathbfcal{L}\in\mathfrak{C}, \quad \lambda\in[0,1], \quad \mathbfcal{N}: {\rm vector\;of\;channels}, \quad
\lambda\mathbfcal{E} + (1-\lambda)\mathbfcal{N} = \mathbfcal{L},
\end{aligned}
\end{eqnarray}
where the last condition holds if and only if $\lambda\mE_{\rm XX'}^\mathcal{J} + (1-\lambda)\mathcal{N}_{\rm XX'}^\mathcal{J} = \mathcal{L}_{\rm XX'}^\mathcal{J}$ $\CY{\forall\;{\rm X|X'}\in\Lambda}$.
By Lemma~\ref{Lemma:Characterization}, there exists a global channel $\CY{\mathcal{L}_{\rm S|S'}}$ such that ${\rm tr}_{\rm S\setminus X}\left(\mathcal{L}_{\rm SS'}^\mathcal{J}\right) = \mathcal{L}_{\rm XX'}^\mathcal{J}\otimes\frac{\id_{{\rm S}'\setminus{\rm X}'}}{d_{{\rm S}'\setminus{\rm X}'}}$ $\CY{\forall\;{\rm X|X'}\in\Lambda}$, which implies ${\rm tr}_{\rm S\setminus X}\left(\mathcal{L}_{\rm SS'}^\mathcal{J}\right) = \left[\lambda\mE_{\rm XX'}^\mathcal{J} + (1-\lambda)\mathcal{N}_{\rm XX'}^\mathcal{J}\right]\otimes\frac{\id_{{\rm S}'\setminus{\rm X}'}}{d_{{\rm S}'\setminus{\rm X}'}}\;\forall\;\CY{{\rm X|X'}\in\Lambda.}$
This means, for every $\CY{{\rm X|X'}\in\Lambda}$,
\begin{align}
&{\rm tr}_{\rm S\setminus X}\left(\mathcal{L}_{\rm SS'}^\mathcal{J}\right)\ge\lambda\mE_{\rm XX'}^\mathcal{J}\otimes\frac{\id_{{\rm S}'\setminus{\rm X}'}}{d_{{\rm S}'\setminus{\rm X}'}};\nonumber\\
&{\rm tr}_{{\rm S}\setminus{\rm X}}\left(\mathcal{L}_{\rm SS'}^\mathcal{J}\right) = {\rm tr}_{{\rm SS'}\setminus{\rm XX'}}\left(\mathcal{L}_{\rm SS'}^\mathcal{J}\right)\otimes\frac{\id_{{\rm S'}\setminus{\rm X'}}}{d_{{\rm S'}\setminus{\rm X'}}};\nonumber\\
&{\rm tr}_{\rm S}\left(\mathcal{L}_{\rm SS'}^\mathcal{J}\right) = {\rm tr}_{\rm X}\left(\mathcal{L}_{\rm XX'}^\mathcal{J}\otimes\frac{\id_{{\rm S}'\setminus{\rm X}'}}{d_{{\rm S}'\setminus{\rm X}'}}\right) = \frac{\id_{\rm S'}}{d_{\rm S'}}.
\end{align}
Hence, when $(\mathbfcal{N},\mathbfcal{L},\lambda)$ is feasible for Eq.~\eqref{Eq:AppMaxForm}, $\left(\mathcal{L}_{\rm SS'}^\mathcal{J},\lambda\right)$ is feasible for Eq.~\eqref{Eq:Result:Conjecture2} in the main text.

Conversely, if $(\rho_{\rm SS'},\lambda)$ is feasible for Eq.~\eqref{Eq:Result:Conjecture2}, then the state $\rho_{\rm SS'} = \mathcal{L}_{\rm SS'}^\mathcal{J}$ is a Choi state of a global channel $\CY{\mathcal{L}_{\rm S|S'}(\cdot)\coloneqq d_{\rm S'}}$ $\CY{{\rm tr}_{\rm S'}\left\{[\id_{\rm S}\otimes(\cdot)^T]\rho_{\rm SS'}\right\}}$ [using the first condition in Eq.~\eqref{Eq:Result:Conjecture2}, we have $\tr_{{\rm S}}(\rho_{\rm SS'}) = \frac{\id_{\rm S'}}{d_{\rm S'}}$].
By Lemma~\ref{Lemma:Characterization} and the last condition in \CY{Eq.~(9), $\mathcal{L}_{\rm S|S'}$} has a well-defined marginal in each \CY{${\rm X|X'}\in\Lambda$}, denoted by \CY{$\mathcal{L}_{\rm X|X'}$}, whose Choi state is given by 
\begin{align}\label{Eq:ComputationStep1}
\mathcal{L}_{\rm XX'}^\mathcal{J} = {\rm tr}_{{\rm SS'}\setminus{\rm XX'}}\left(\rho_{\rm SS'}\right).
\end{align}
Now, recall the second condition in Eq.~\eqref{Eq:Result:Conjecture2}; namely, $\lambda \mE_{{\rm X}{\rm X}'}^\mathcal{J}\otimes\frac{\id_{{\rm S}'\setminus{\rm X}'}}{d_{{\rm S}'\setminus{\rm X}'}}\le\tr_{{\rm S}\setminus{\rm X}}(\rho_{\rm SS'})$.
Tracing out ${\rm S'\setminus X'}$ and using Eq.~\eqref{Eq:ComputationStep1}, we obtain the inequality 
\begin{align}\label{Eq:ComputationStep2-1}
\mathcal{L}_{\rm XX'}^\mathcal{J} - \lambda \mE_{{\rm X}{\rm X}'}^\mathcal{J}\ge0.
\end{align}
On the other hand, since $\mathcal{L}_{\rm XX'}^\mathcal{J}$ and $\mathcal{E}_{\rm XX'}^\mathcal{J}$ both have $\frac{\id_{\rm X'}}{d_{\rm X'}}$ as their marginals in ${\rm X'}$, we learn that, when $\lambda<1$, 
\begin{align}\label{Eq:ComputationStep2-2}
\frac{1}{1-\lambda}{\rm tr}_{\rm X}\left(\mathcal{L}_{\rm XX'}^\mathcal{J} - \lambda \mE_{{\rm X}{\rm X}'}^\mathcal{J}\right) = \frac{\id_{\rm X'}}{d_{\rm X'}}.
\end{align}
Eqs.~\eqref{Eq:ComputationStep2-1} and~\eqref{Eq:ComputationStep2-2} imply that, for $\lambda<1$, $\frac{1}{1-\lambda}\left(\mathcal{L}_{\rm XX'}^\mathcal{J} - \lambda \mE_{{\rm X}{\rm X}'}^\mathcal{J}\right)$ is a legal Choi state.
In other words, \CY{$\frac{1}{1-\lambda}\left(\mathcal{L}_{\rm X|X'} - \lambda\mathcal{E}_{\rm X|X'}\right)$ is a channel from ${\rm X'}$ to ${\rm X}$} when $\lambda<1$.
For each \CY{${\rm X|X'}\in\Lambda$}, by defining the channel
\begin{align}
\CY{\mathcal{N}_{\rm X|X'}\coloneqq\frac{1}{1-\lambda}\left(\mathcal{L}_{\rm X|X'} - \lambda\mathcal{E}_{\rm X|X'}\right)}\quad{\rm if}\;\lambda<1\quad\&\quad\CY{\mathcal{N}_{\rm X|X'}:\CYtwo{\rm arbitrary}}\quad{\rm if}\;\lambda = 1,
\end{align}
we have \CY{$\mathcal{L}_{\rm X|X'} = \lambda\mathcal{E}_{\rm X|X'} + (1-\lambda)\mathcal{N}_{\rm X|X'}$ $\forall\;{\rm X|X'}\in\Lambda$, showing that $\left(\{\mathcal{N}_{\rm X|X'}\}_{\rm X|X'\in\Lambda},\{\mathcal{L}_{\rm X|X'}\}_{\rm X|X'\in\Lambda},\lambda\right)$} is feasible for Eq.~\eqref{Eq:AppMaxForm} for every $(\rho_{\rm SS'},\lambda)$ that is feasible for Eq.~\eqref{Eq:Result:Conjecture2}.
Thus, the two optimization problems Eqs.~\eqref{Eq:AppMaxForm} and~\eqref{Eq:Result:Conjecture2} have the same optimum.
\end{proof}

\subsection{Dual Problem of Eq.~(\ref{Eq:Result:Conjecture2})}\label{App:Dual}
\CYtwo{
In what follows, when an operator is written as $V^{\rm (X)}_{\rm Y|Z}$, it means that it is an operator acting on the system ${\rm X}$ with dependency on the \CYthree{output-input pair} ${\rm Y|Z}$.
}
\begin{alemma}\label{Lemma:DualSDP}
The dual of Eq.~\eqref{Eq:Result:Conjecture2} is
\begin{eqnarray}\label{Eq:DualSDP}
\begin{aligned}
\min_{z,H_{\rm S'},\CYtwo{\left\{\HXS\right\},\left\{\ZXS\right\}}}\quad&\frac{{\rm tr}(H_{\rm S'})}{d_{\rm S'}} + z\\
{\rm s.t.}\quad&\id_{\rm S}\otimes H_{\rm S'} + \sum_\CYtwo{\rm X|X'\in\Lambda}\left(\CYtwo{\HXS} - {\rm tr}_{\rm S'\setminus X'}\left(\CYtwo{\HXS}\right)\otimes\frac{\id_{\rm S'\setminus X'}}{d_{\rm S'\setminus X'}} - \CYtwo{\ZXS}\right)\otimes\id_{\rm S\setminus X}\ge0\\
&z + \sum_\CYtwo{\rm X|X'\in\Lambda}{\rm tr}\left[\CYtwo{\ZXS}\left(\mathcal{E}_{\rm XX'}^\mathcal{J}\otimes\frac{\id_{\rm S'\setminus X'}}{d_{\rm S'\setminus X'}}\right)\right]\ge1\\
&\CYtwo{\HXSd = \HXS}\;\CY{\rm \forall\;X|X'\in\Lambda}, \quad H_{\rm S'}^\dagger = H_{\rm S'}; \quad \CYtwo{\ZXS}\ge0\;\CY{\rm \forall\;X|X'\in\Lambda}, \quad z\ge0.
\end{aligned}
\end{eqnarray}
Furthermore, strong duality holds; hence, Eqs.~\eqref{Eq:Result:Conjecture2} and~\eqref{Eq:DualSDP} have the same optimum.
\end{alemma}
\begin{proof}
First, we write Eq.~\eqref{Eq:Result:Conjecture2} in the standard form.
In what follows, $\mathcal{L}({\rm X})$ denotes the set of linear maps on ${\rm X}$.
Then let $V \coloneqq \rho_{\rm SS'}\oplus\lambda$ and $A\coloneqq 0\oplus 1$, both in $\mathcal{L}({\rm SS'})\oplus\mathbb{R}$, such that $\langle V,A\rangle\coloneqq{\rm tr}(V^\dagger A) = \lambda$.
Note that $x\oplus y\coloneqq\left(\begin{matrix}x &\cdot\\ \cdot&y\end{matrix}\right)$, where the off-diagonal terms are irrelevant to the definition of the direct sum spaces.
This is the objective function for the standard form that we will adopt.
The feasible set can be characterized by defining the following functions:
\CYtwo{
$
\Phi\coloneqq\left(\bigoplus_{\rm X|X'\in\Lambda}\Phi_{\rm X|X'}\right)\oplus\Phi_0$ and $
\Psi\coloneqq\left(\bigoplus_{\rm X|X'\in\Lambda}\Psi_{\rm X|X'}\right)\oplus\Psi_0,
$
}
where for each $\CY{\rm X|X'\in\Lambda}$ we have
\begin{align}
&\Phi_\CY{{\rm X|X'}}(V)\coloneqq{\rm tr}_{\rm S\setminus X}(\rho_{\rm SS'}) -{\rm tr}_{\rm SS'\setminus XX'}\left(\rho_{\rm SS'}\right)\otimes\frac{\id_{\rm S'\setminus X'}}{d_{\rm S'\setminus X'}};\\
&\Phi_0(V)\coloneqq{\rm tr}_{\rm S}(\rho_{\rm SS'});\\
&\Psi_\CY{{\rm X|X'}}(V)\coloneqq\lambda\mathcal{E}_{\rm XX'}^\mathcal{J}\otimes\frac{\id_{\rm S'\setminus X'}}{d_{\rm S'\setminus X'}} - {\rm tr}_{\rm S\setminus X}\left(\rho_{\rm SS'}\right);\\
&\Psi_0(V)\coloneqq\lambda,
\end{align}
\CY{where $\Phi_{\rm X|X'}:\mathcal{L}({\rm SS'})\oplus\mathbb{R}\to\mathcal{L}({\rm XS'})$, $\Phi_0:\mathcal{L}({\rm SS'})\oplus\mathbb{R}\to\mathcal{L}({\rm S'})$, $\Psi_{\rm X|X'}:\mathcal{L}({\rm SS'})\oplus\mathbb{R}\to\mathcal{L}({\rm XS'})$, and $\Psi_0:\mathcal{L}({\rm SS'})\oplus\mathbb{R}\to\mathbb{R}$.}
One can check that all of them are hermitian-preserving linear maps.
Also, by choosing 
\begin{align}
&B\coloneqq\left(\bigoplus_\CY{{\rm X|X'\in\Lambda}}0_{\rm XS'}\right)\oplus \frac{\id_{\rm S'}}{d_{\rm S'}};
\quad \quad C\coloneqq\left(\bigoplus_\CY{{\rm X|X'\in\Lambda}}0_{\rm XS'}\right)\oplus1,
\end{align}
where $0_{\rm XS'}$ denotes the zero operator in the system ${\rm XS'}$.
\CYtwo{Note that this means for two different \CYthree{output-input pairs} ${\rm X|X'},{\rm Y|Y'}$, we attribute the same zero operator to them if ${\rm X=Y}$.}
One can rewrite Eq.~\eqref{Eq:Result:Conjecture2} as
\begin{eqnarray}\label{Eq:SDPStandardForm}
\begin{aligned}
\max_{V}\quad&\langle V,A\rangle\\
{\rm s.t.}\quad&\Phi(V) = B, \quad 
\Psi(V)\le C, \quad V\ge0.
\end{aligned}
\end{eqnarray}
The corresponding dual problem is (see, e.g., Sec. 1.2.3 in Ref.~\cite{Watrous-Book})
\begin{eqnarray}\label{Eq:DualSDPStandardForm}
\begin{aligned}
\min_{H,Z}\quad&\langle H,B\rangle + \langle Z,C\rangle\\
{\rm s.t.}\quad&\Phi^\dagger(H) + \Psi^\dagger(Z) \ge A, \quad H^\dagger = H, \quad Z\ge0,
\end{aligned}
\end{eqnarray}
where 
\CYtwo{$
H = \left(\bigoplus_{\rm X|X'\in\Lambda}\HXS\right)\oplus H_{\rm S'}\in\left(\bigoplus_{\rm X|X'\in\Lambda}\mathcal{L}({\rm XS'})\right)\oplus\mathcal{L}({\rm S'})$ and $
Z = \left(\bigoplus_{\rm X|X'\in\Lambda}\ZXS\right)\oplus z\in\left(\bigoplus_{\rm X|X'\in\Lambda}\mathcal{L}({\rm XS'})\right)\oplus\mathbb{R}.
$
}
Now it remains to find $\Phi^\dagger$ and $\Psi^\dagger$.
First, we note that
\begin{align}
\langle \Phi^\dagger(H),V\rangle &= \left\langle  \left(\bigoplus_\CY{{\rm X|X'\in\Lambda}}\CYtwo{\HXS}\right)\oplus H_{\rm S'}, \left(\bigoplus_\CY{{\rm X|X'\in\Lambda}}\Phi_\CY{{\rm X|X'}}(V)\right)\oplus\Phi_0(V)\right\rangle\nonumber\\
&=\langle H_{\rm S'},\Phi_0(V)\rangle + \sum_\CY{{\rm X|X'\in\Lambda}}\left\langle \CYtwo{\HXS},\CY{\Phi_{\rm X|X'}}(V)\right\rangle = \left\langle \Phi_0^\dagger(H_{\rm S'}) + \sum_\CY{{\rm X|X'\in\Lambda}}\CY{\Phi_{\rm X|X'}^\dagger}\left(\CYtwo{\HXS}\right),V\right\rangle,
\end{align}
which means $\Phi^\dagger(H) = \Phi_0^\dagger(H_{\rm S'}) + \CY{\sum_{\rm X|X'\in\Lambda}\Phi_{\rm X|X'}^\dagger}\left(\CYtwo{\HXS}\right)$, \CYtwo{and, similarly, $\Psi^\dagger(Z) = \Psi_0^\dagger(z) + \sum_{\rm X|X'\in\Lambda}\Psi_{\rm X|X'}^\dagger\left(\ZXS\right)$.}
So we find the adjoint of each $\CY{\Phi_{\rm X|X'}},\Phi_0,\CY{\Psi_{\rm X|X'}},\Psi_0$ separately.
Direct computation shows $\langle\Phi_0^\dagger(H_{\rm S'}),V\rangle = {\rm tr}\left[H_{\rm S'}{\rm tr}_{\rm S}\left(\rho_{\rm SS'}\right)\right] = \langle \left(\id_{\rm S}\otimes H_{\rm S'}\right)\oplus0,V\rangle,$ which means 
$
\Phi_0^\dagger(H_{\rm S'}) = \left(\id_{\rm S}\otimes H_{\rm S'}\right)\oplus0.
$
For a given $\CY{\rm X|X'\in\Lambda}$, we have
\begin{align}
\langle\CY{\Phi_{\rm X|X'}^\dagger}\left(\CYtwo{\HXS}\right),V\rangle &= \left\langle \CYtwo{\HXS},{\rm tr}_{\rm S\setminus X}(\rho_{\rm SS'}) -{\rm tr}_{\rm SS'\setminus XX'}\left(\rho_{\rm SS'}\right)\otimes\frac{\id_{\rm S'\setminus X'}}{d_{\rm S'\setminus X'}}\right\rangle\nonumber\\
& ={\rm tr}\left[(\CYtwo{\HXS}\otimes\id_{\rm S\setminus X})\rho_{\rm SS'}\right] - {\rm tr}\left[\left(\frac{{\rm tr}_{\rm S'\setminus X'}\left(\CYtwo{\HXS}\right)}{d_{\rm S'\setminus X'}}\otimes\id_{\rm SS'\setminus XX'}\right)\rho_{\rm SS'}\right]\nonumber\\
&=\left\langle \left(\CYtwo{\HXS}\otimes\id_{\rm S\setminus X} - \frac{{\rm tr}_{\rm S'\setminus X'}\left(\CYtwo{\HXS}\right)}{d_{\rm S'\setminus X'}}\otimes\id_{\rm SS'\setminus XX'}\right)\oplus 0,V\right\rangle,
\end{align}
which implies
\CYtwo{
$
\CY{\Phi_{\rm X|X'}^\dagger}(\CYtwo{\HXS}) = \left[\left(\CYtwo{\HXS} - {\rm tr}_{\rm S'\setminus X'}\left(\CYtwo{\HXS}\right)\otimes\frac{\id_{\rm S'\setminus X'}}{d_{\rm S'\setminus X'}}\right)\otimes\id_{\rm S\setminus X}\right]\oplus 0.
$
}
On the other hand, we have $\langle\Psi_0^\dagger(z),V\rangle = z\lambda$ and hence 
$
\Psi_0^\dagger(z) = 0\oplus z.
$
Also, for a given $\CY{\rm X|X'\in\Lambda}$,
\begin{align}
\langle\CY{\Psi_{\rm X|X'}^\dagger}(\CYtwo{\ZXS}),V\rangle &= \left\langle \CYtwo{\ZXS},\lambda\mathcal{E}_{\rm XX'}^\mathcal{J}\otimes\frac{\id_{\rm S'\setminus X'}}{d_{\rm S'\setminus X'}} - {\rm tr}_{\rm S\setminus X}\left(\rho_{\rm SS'}\right)\right\rangle\nonumber\\
& =\lambda{\rm tr}\left[\CYtwo{\ZXS}\left(\mathcal{E}_{\rm XX'}^\mathcal{J}\otimes\frac{\id_{\rm S'\setminus X'}}{d_{\rm S'\setminus X'}}\right)\right] - {\rm tr}\left[\left(\CYtwo{\ZXS}\otimes\id_{\rm S\setminus X}\right)\rho_{\rm SS'}\right]\nonumber\\
&=\left\langle \left(-\CYtwo{\ZXS}\otimes\id_{\rm S\setminus X}\right)\oplus {\rm tr}\left[\CYtwo{\ZXS}\left(\mathcal{E}_{\rm XX'}^\mathcal{J}\otimes\frac{\id_{\rm S'\setminus X'}}{d_{\rm S'\setminus X'}}\right)\right],V\right\rangle.
\end{align}
From here we conclude that
\CYtwo{
$
\CY{\Psi_{\rm X|X'}^\dagger}(\CYtwo{\ZXS}) = \left(-\CYtwo{\ZXS}\otimes\id_{\rm S\setminus X}\right)\oplus {\rm tr}\left[\CYtwo{\ZXS}\left(\mathcal{E}_{\rm XX'}^\mathcal{J}\otimes\frac{\id_{\rm S'\setminus X'}}{d_{\rm S'\setminus X'}}\right)\right].
$
}

Combining everything and replacing in Eq.~\eqref{Eq:DualSDPStandardForm}, one can obtain Eq.~\eqref{Eq:DualSDP}.
More precisely, we have $\langle H,B\rangle = {\rm tr}\left(H_{\rm S'}\frac{\id_{\rm S'}}{d_{\rm S'}}\right), \langle Z,C\rangle = z$, and
\begin{align}
\Phi^\dagger(H) + \Psi^\dagger(Z) = &\left[\id_{\rm S}\otimes H_{\rm S'} + \sum_\CY{{\rm X|X'\in\Lambda}}\left(\CYtwo{\HXS} - {\rm tr}_{\rm S'\setminus X'}\left(\CYtwo{\HXS}\right)\otimes\frac{\id_{\rm S'\setminus X'}}{d_{\rm S'\setminus X'}}-\CYtwo{\ZXS}\right)\otimes\id_{\rm S\setminus X}\right]\nonumber\\
&\oplus\left(z + \sum_\CY{{\rm X|X'\in\Lambda}}{\rm tr}\left[\CYtwo{\ZXS}\left(\mathcal{E}_{\rm XX'}^\mathcal{J}\otimes\frac{\id_{\rm S'\setminus X'}}{d_{\rm S'\setminus X'}}\right)\right]\right).
\end{align}
Finally, note that the primal problem is finite and feasible (taking $\rho_{\rm SS'} = \frac{\id_{\rm SS'}}{d_{\rm SS'}}$ and $\lambda = 0$). 
Also, the dual is strictly feasible by taking any \CYtwo{$\HXS$} Hermitian, \CYtwo{$\ZXS > 0$}, and $H_{\rm S'} = h\id_{\rm S'}$ for $h>0, z>0$ large enough. 
By Slater's condition (Theorem 1.18 in Ref.~\cite{Watrous-Book}), strong duality holds.
\end{proof}

\subsection{Proof of Eq.~(\ref{Eq:Result:CompatibilityProblemCondition})}
From Lemma~\ref{Lemma:DualSDP} we can prove the witness form [Eq.~\eqref{Eq:Result:CompatibilityProblemCondition} in the main text] as follows:
\begin{proof}
Since the validity of Eq.~\eqref{Eq:Result:CompatibilityProblemCondition} implies that $\mathbfcal{E}$ is incompatible, it suffices to show the necessity; i.e., incompatibility implies Eq.~\eqref{Eq:Result:CompatibilityProblemCondition}.
To start with, we note that when $\left\{\CYtwo{\ZXS}\right\}_\CY{{\rm X|X'\in\Lambda}}$ and $z$ satisfy \CYtwo{the second constraint} in Eq.~\eqref{Eq:DualSDP}, we have $z\ge 1-\sum_\CY{{\rm X|X'\in\Lambda}}{\rm tr}\left[\CYtwo{\ZXS}\left(\mathcal{E}_{\rm XX'}^\mathcal{J}\otimes\frac{\id_{\rm S'\setminus X'}}{d_{\rm S'\setminus X'}}\right)\right]$.
Using Lemma~\ref{Lemma:DualSDP} and strong duality, we conclude that $R(\mathbfcal{E})$ is lower bounded by
\begin{eqnarray}\label{Eq:LowerBound-Computation-DualSDP2}
\begin{aligned}
\min_{H_{\rm S'},\CYtwo{\left\{\HXS\right\},\left\{\ZXS\right\}}}\quad&\frac{{\rm tr}(H_{\rm S'})}{d_{\rm S'}} + 1 - \sum_\CY{{\rm X|X'\in\Lambda}}{\rm tr}\left[\CYtwo{\ZXS}\left(\mathcal{E}_{\rm XX'}^\mathcal{J}\otimes\frac{\id_{\rm S'\setminus X'}}{d_{\rm S'\setminus X'}}\right)\right]\\
{\rm s.t.}\quad&\id_{\rm S}\otimes H_{\rm S'} + \sum_\CY{{\rm X|X'\in\Lambda}}\left(\CYtwo{\HXS} - {\rm tr}_{\rm S'\setminus X'}\left(\CYtwo{\HXS}\right)\otimes\frac{\id_{\rm S'\setminus X'}}{d_{\rm S'\setminus X'}} - \CYtwo{\ZXS}\right)\otimes\id_{\rm S\setminus X}\ge0\\
& \CYtwo{\HXSd = \HXS}\;\CY{\rm \forall\;X|X'\in\Lambda}, \quad H_{\rm S'}^\dagger = H_{\rm S'},  \quad \CYtwo{\ZXS}\ge0\;\CY{\rm \forall\;X|X'\in\Lambda}.
\end{aligned}
\end{eqnarray}
Note that since the objective function becomes independent of $z$, so \CYtwo{the second constraint} in Eq.~\eqref{Eq:DualSDP} always holds (e.g., by a given $\left\{\CYtwo{\ZXS}\right\}_\CY{{\rm X|X'\in\Lambda}}$ and a high enough finite $z$) and hence can be dropped.
Now we consider the following estimate
\begin{align}
\frac{{\rm tr}(H_{\rm S'})}{d_{\rm S'}} &= \max_{\mathbfcal{L}\in\mathfrak{C}}{\rm tr}\left[(\id_{\rm S}\otimes H_{\rm S'})\mathcal{L}_{\rm SS'}^\mathcal{J}\right]\ge\max_{\mathbfcal{L}\in\mathfrak{C}}\sum_\CY{{\rm X|X'\in\Lambda}}{\rm tr}\left[\left(\left(-\CYtwo{\HXS} + {\rm tr}_{\rm S'\setminus X'}\left(\CYtwo{\HXS}\right)\otimes\frac{\id_{\rm S'\setminus X'}}{d_{\rm S'\setminus X'}} + \CYtwo{\ZXS}\right)\otimes\id_{\rm S\setminus X}\right)\mathcal{L}_{\rm SS'}^\mathcal{J}\right]\nonumber\\
&=\max_{\mathbfcal{L}\in\mathfrak{C}}\sum_\CY{{\rm X|X'\in\Lambda}}{\rm tr}\left[\left(-\CYtwo{\HXS} + {\rm tr}_{\rm S'\setminus X'}\left(\CYtwo{\HXS}\right)\otimes\frac{\id_{\rm S'\setminus X'}}{d_{\rm S'\setminus X'}} + \CYtwo{\ZXS}\right)\left(\mathcal{L}_{\rm XX'}^\mathcal{J}\otimes\frac{\id_{\rm S'\setminus X'}}{d_{\rm S'\setminus X'}}\right)\right]\nonumber\\
&=\max_{\mathbfcal{L}\in\mathfrak{C}}\sum_\CY{{\rm X|X'\in\Lambda}}{\rm tr}\left[\CYtwo{\ZXS}\left(\mathcal{L}_{\rm XX'}^\mathcal{J}\otimes\frac{\id_{\rm S'\setminus X'}}{d_{\rm S'\setminus X'}}\right)\right],
\end{align}
where in the first line $\mathcal{L}_{\rm SS'}^\mathcal{J}$ denotes the Choi state of the global channel $\CY{\mathcal{L}_{\rm S|S'}}$ compatible with $\mathbfcal{L} = \CY{\{\mathcal{L}_{\rm X|X'}\}_{\rm X|X'\in\Lambda}}$, and the equality follows from the property of a Choi state; the \CY{inequality} is due to \CYtwo{the first constraint} in Eq.~\eqref{Eq:DualSDP}; \CY{the second line} is due to Lemma~\ref{Lemma:Characterization}; the last line is because of ${\rm tr}\left[\CYtwo{\HXS}\left(\mathcal{L}_{\rm XX'}^\mathcal{J}\otimes\frac{\id_{\rm S'\setminus X'}}{d_{\rm S'\setminus X'}}\right)\right] = {\rm tr}\left[{\rm tr}_{\rm S'\setminus X'}\left(\CYtwo{\HXS}\right)\otimes\frac{\id_{\rm S'\setminus X'}}{d_{\rm S'\setminus X'}} \left(\mathcal{L}_{\rm XX'}^\mathcal{J}\otimes\frac{\id_{\rm S'\setminus X'}}{d_{\rm S'\setminus X'}}\right)\right]$.
From here we conclude that $R(\mathbfcal{E})$ is lower bounded by
\begin{eqnarray}
\begin{aligned}
\min_{\left\{\CYtwo{\ZXS}\right\}}\quad&1 - \sum_\CY{{\rm X|X'\in\Lambda}}{\rm tr}\left[\CYtwo{\ZXS}\left(\mathcal{E}_{\rm XX'}^\mathcal{J}\otimes\frac{\id_{\rm S'\setminus X'}}{d_{\rm S'\setminus X'}}\right)\right] + \max_{\mathbfcal{L}\in\mathfrak{C}}\sum_\CY{{\rm X|X'\in\Lambda}}{\rm tr}\left[\CYtwo{\ZXS}\left(\mathcal{L}_{\rm XX'}^\mathcal{J}\otimes\frac{\id_{\rm S'\setminus X'}}{d_{\rm S'\setminus X'}}\right)\right]\\
{\rm s.t.}\quad&\CYtwo{\ZXS}\ge0\;\CY{\rm \forall\;X|X'\in\Lambda}.
\end{aligned}
\end{eqnarray}
The first constraint in Eq.~\eqref{Eq:LowerBound-Computation-DualSDP2} is dropped because it always holds with a properly chosen $H_{\rm S'}$ (e.g., one can choose $H_{\rm S'} = h\id_{\rm S'}$ with a high enough $h$) and the objective function is independent of $H_{\rm S'},\left\{\CYtwo{\HXS}\right\}_\CY{{\rm X|X'\in\Lambda}}$.
This means that when $\mathbfcal{E}$ is incompatible, and hence $R(\mathbfcal{E})<1$, there exist positive operators $\CYtwo{\ZXS}$ such that
\CYtwo{
$
\max_{\mathbfcal{L}\in\mathfrak{C}}\sum_\CY{{\rm X|X'\in\Lambda}}{\rm tr}\left[\CYtwo{\ZXS}\left(\mathcal{L}_{\rm XX'}^\mathcal{J}\otimes\frac{\id_{\rm S'\setminus X'}}{d_{\rm S'\setminus X'}}\right)\right] < \sum_\CY{{\rm X|X'\in\Lambda}}{\rm tr}\left[\CYtwo{\ZXS}\left(\mathcal{E}_{\rm XX'}^\mathcal{J}\otimes\frac{\id_{\rm S'\setminus X'}}{d_{\rm S'\setminus X'}}\right)\right].
$
}
Finally, let $\CYtwo{\wt{H}_{\rm X|X'}}\coloneqq\frac{{\rm tr}_{\rm S'\setminus X'}\left(\CYtwo{\ZXS}\right)}{d_{\rm S'\setminus X'}}$, which is again positive.
Then the above inequality implies 
\CYtwo{
$
\max_{\mathbfcal{L}\in\mathfrak{C}}\sum_\CY{{\rm X|X'\in\Lambda}}{\rm tr}\left(\CYtwo{\wt{H}_{\rm X|X'}}\mathcal{L}_{\rm XX'}^\mathcal{J}\right) < \sum_\CY{{\rm X|X'\in\Lambda}}{\rm tr}\left(\CYtwo{\wt{H}_{\rm X|X'}}\mE_{\rm XX'}^\mathcal{J}\right),
$
}
\CYtwo{and the proof is completed.}
\end{proof}

\section{Examples of Incompatibility}\label{App:Example}
In this section, we go through the computational details for the examples of incompatibility that we mention in the main text.
We start with constructing a family of examples in a simple tripartite setting ${\rm ABC}$ with ${\rm A\simeq A'\simeq C\simeq C'}$ and ${\rm B\simeq B'}$.
Let $\mE_{\rm X} = \mE_{\rm X|X'}$ when ${\rm X=X'}$.
In what follows, $\ket{\phi_{\rm XBB'}}$ is a pure state satisfying ${\rm tr}_{\rm XB}(\proj{\phi_{\rm XBB'}}) = \frac{\id_{\rm B'}}{d_{\rm B'}}$.
We use subscripts to indicate the same state distributed among different local systems.
Define a channel $\mathcal{M}_{\rm XB}$ through its Choi state ({\rm X = A,C}):
\begin{align}\label{Eq:ExampleOne}
\mathcal{M}_{\rm XX'BB'}^\mathcal{J}\coloneqq\proj{\phi_{\rm XBB'}}\otimes\frac{\id_{\rm X'}}{d_{\rm X'}},
\end{align}
It is easy to see that $\mathcal{M}_{\rm AB}$ and $\mathcal{M}_{\rm CB}$ are locally compatible; i.e., $\Tr_{\rm A|A'}(\mathcal{M}_{\rm AB})=\Tr_{\rm C|C'}(\mathcal{M}_{\rm CB})$. 
Taking the Hermitian operators $H_{\rm XB|X'B'} = d_{\rm X'}\mathcal{M}_{\rm XX'BB'}^\mathcal{J}$ and using Eq.~\eqref{Eq:Result:CompatibilityProblemCondition} in the main text, one can show the following result:
\newpage
\begin{aproposition}
$\mathcal{M}_{\rm AB}$ and $\mathcal{M}_{\rm CB}$ are incompatible if and only if $\ket{\phi_{\rm XBB'}}$ is non-product in the $\CYthree{\rm X\;vs.\;BB'}$ bipartition.
\end{aproposition}
\begin{proof}
First, the sufficiency ($\Rightarrow$) can be seen by showing the counterpositive.
Suppose $\ket{\phi_{\rm XBB'}} = \ket{\phi}_{\rm X}\otimes \ket{\xi}_{\rm BB'}$, which is product in $\CYthree{\rm X\;vs.\;BB'}$ bipartition.
Then by Lemma~\ref{Lemma:Characterization} (${\rm S = ABC}$) $\mE_{\rm SS'}^\mathcal{J} = \proj{\phi}_{\rm A}\otimes\proj{\phi}_{\rm C}\otimes\proj{\xi}_{\rm BB'}\otimes\frac{\id_{\rm A'C'}}{d_{\rm A'C'}}$ is the Choi state of a global channel compatible with $\{\mathcal{M}_{\rm AB},\mathcal{M}_{\rm CB}\}$.

To prove the necessity ($\Leftarrow$), consider the Hermitian operators $H_{\rm XB|X'B'}\coloneqq d_{\rm X'}\mathcal{M}_{\rm XX'BB'}^\mathcal{J}$ (${\rm X = A,C}$).
Then we have
\begin{align}
\CYtwo{\max_{\mathbfcal{L}\in\mathfrak{C}}}\left[{\rm tr}\left(\CY{H_{\rm AB|A'B'}}\mathcal{L}_{\rm AA'BB'}^\mathcal{J}\right) + {\rm tr}\left(\CY{H_{\rm CB|C'B'}}\mathcal{L}_{\rm CC'BB'}^\mathcal{J}\right)\right]&\CYtwo{\le}\CYtwo{\max_{\rho_{\rm SS'}}}\;{\rm tr}\left[\left(\proj{\phi_{\rm ABB'}}\otimes\id_{\rm A'CC'}+\proj{\phi_{\rm CBB'}}\otimes\id_{\rm AA'C'}\right)\rho_{\rm SS'}\right]\nonumber\\
&\CYtwo{\le 2},
\end{align}
where the last inequality is saturated if and only if there exists a state $\rho_{\rm ABB'C}$ such that
$
{\rm tr}\left[\left(\proj{\phi_{\rm ABB'}}\otimes\id_{\rm C}\right)\rho_{\rm ABB'C}\right] = 1 = {\rm tr}\left[\left(\proj{\phi_{\rm CBB'}}\otimes\id_{\rm A}\right)\rho_{\rm ABB'C}\right].
$
This is true if and only if ${\rm tr}_{\rm C}(\rho_{\rm ABB'C}) = \proj{\phi_{\rm ABB'}}$ and ${\rm tr}_{\rm A}(\rho_{\rm ABB'C}) = \proj{\phi_{\rm CBB'}}$.
Now, by entanglement monogamy (see, e.g., Ref.~\cite{Doherty2014}), it is impossible for such $\rho_{\rm ABB'C}$ to exist when $\ket{\phi_{\rm XBB'}}$ is non-product in \CYthree{${\rm X\;vs.\;BB'}$ bipartition}. 
Since, by assumption, $\ket{\phi_{\rm XBB'}}$ is non-product, we conclude that
\CYtwo{
$
\CYtwo{\max_{\mathbfcal{L}\in\mathfrak{C}}}\left[{\rm tr}\left(\CY{H_{\rm AB|A'B'}}\mathcal{L}_{\rm AA'BB'}^\mathcal{J}\right) + {\rm tr}\left(\CY{H_{\rm CB|C'B'}}\mathcal{L}_{\rm CC'BB'}^\mathcal{J}\right)\right] \CYtwo{< 2}.
$
}
On the other hand, a direct computation shows that 
$
{\rm tr}\left(\CY{H_{\rm AB|A'B'}}\mathcal{M}_{\rm AA'BB'}^\mathcal{J}\right) + {\rm tr}\left(\CY{H_{\rm CB|C'B'}}\mathcal{M}_{\rm CC'BB'}^\mathcal{J}\right) \CYtwo{= 2}.
$
By Eq.~\eqref{Eq:Result:CompatibilityProblemCondition} in the main text we learn that $\{\mathcal{M}_{\rm AB},\mathcal{M}_{\rm CB}\}\notin\mathfrak{C}$; i.e., incompatible.
\end{proof}
A potential physical explanation for this result is that, due to entanglement monogamy~\cite{Coffman2000}, it is impossible to clone quantum correlations; i.e., the mapping $\ket{\psi_{\rm AB}}\to\ket{\psi_{\rm ABC}}$ such that $\ket{\psi_{\rm AB}} = \ket{\psi_{\rm BC}}$ is possible only when $\ket{\psi_{\rm AB}}$ is product in the $\CYthree{\rm A\;vs.\;B}$ bipartition.

\CYthree{
Now we consider the example mentioned in the main text; namely, consider the three-qubit state $\ket{\phi_{\rm XBB'}} = \ket{\rm GHZ_{XBB'}}$, where $\ket{\rm GHZ_{XYZ}}\coloneqq\frac{1}{\sqrt{2}}(\ket{000}_{\rm XYZ} + \ket{111}_{\rm XYZ})$ is the {\em Greenberger-Horne-Zeilinger state}~\cite{GHZ} in a three-qubit system ${\rm XYZ}$.
}
Numerically, we found robustness $0.75$, achieved by the global channel ${\mathcal{G}}_{\rm ABC}$ with the following Choi state:
\begin{align}\label{Eq:NumericalGlobalChannel}
\CYthree{{\mathcal{G}}_{\rm ABCA'B'C'}^\mathcal{J}} = &\frac{\id_{\rm A'C'}}{4}\otimes\frac{1}{12}[4\proj{0000} + \proj{0001} + \ket{0001}\bra{0010} + \ket{0010}\bra{0001} + \proj{0010}\nonumber\\
&+4\proj{1111} + \proj{1101} + \ket{1101}\bra{1110} + \ket{1110}\bra{1101} + \proj{1110}\nonumber\\
&+2\ket{0000}\bra{1101} + 2\ket{0000}\bra{1110} + 2\ket{0001}\bra{1111} + 2\ket{0010}\bra{1111}\nonumber\\
&+2\ket{1101}\bra{0000} + 2\ket{1110}\bra{0000} + 2\ket{1111}\bra{0001} + 2\ket{1111}\bra{0010}]_{\rm BB'AC}.
\end{align} 
\CYthree{Note that, in the above equation and for the rest of this section, subscripts of pure states denote the order of subsystems that the states live in.}
As mentioned in the main text, a natural question is to ask whether Eq.~\eqref{Eq:NumericalGlobalChannel} is from a global channel achieved by cloning the bipartite channel $\mathcal{M}_{\rm XB}$ locally with the help of an optimal universal cloning machine~\cite{Cloning-RMP}.
More precisely, Bob gets an input state in ${\rm B}$, takes an ancillary system ${\rm X}$ initially prepared in $\ket{0}_{\rm X}$ and performs a ${\rm CNOT}_{\rm XB}$ controlled on B. 
Then he applies the optimal universal cloning machine to ${\rm X}$ and sends one copy to Alice and one to Charlie (denote the corresponding channel \CYthree{$\mathcal{C}_{\rm AC|X}$}). 
The resulting channel is 
\begin{align}\label{Eq:CloningProtocol}
\wt{\mathcal{M}}_{\rm ABC}(\cdot)\coloneqq(\CYthree{\mathcal{C}_{\rm AC|X}}\otimes\mathcal{I}_{\rm B})\circ{\rm CNOT}_{\rm XB}[\proj{0}_{\rm X}\otimes{\rm tr}_{\rm AC}(\cdot)],
\end{align}
Formally, an optimal universal cloning machine that clones arbitrary states in ${\rm X}$ into ${\rm AC}$ is a unitary operator acting as~\cite{Cloning-RMP}
\begin{align}\label{Eq:OptimalCloningMachine}
&\ket{0}_{\rm X}\otimes\ket{00}_{\rm CM}\mapsto \sqrt{\frac{2}{3}}\ket{001}_{\rm ACM} - \sqrt{\frac{1}{3}}\ket{\psi_+}_{\rm AC}\otimes\ket{0}_{\rm M};\\
&\ket{1}_{\rm X}\otimes\ket{00}_{\rm CM}\mapsto -\sqrt{\frac{2}{3}}\ket{110}_{\rm ACM} + \sqrt{\frac{1}{3}}\ket{\psi_+}_{\rm AC}\otimes\ket{1}_{\rm M}.
\end{align}
where $\ket{\psi_+}\coloneqq\frac{1}{\sqrt{2}}\left(\ket{10}+\ket{01}\right)$ and ${\rm M}$ is an ancillary system (the ``machine'') that will be dropped in the end.
The map \CY{$\mathcal{C}_{\rm AC|X}$ from ${\rm X}$ to ${\rm AC}$} is then obtained by tracing out ${\rm M}$ after the cloning unitary. To verify that the global channel Eq.~\eqref{Eq:CloningProtocol} is equivalent to the one found numerically, we compute its Choi state, which is given by (note that every single system is a qubit)
\begin{align}\label{Eq:CloningProtocolChoi}
\CYthree{\wt{\mathcal{M}}_{\rm ABCA'B'C'}^\mathcal{J}} = (\CY{\mathcal{C}_{\rm AC|X}}\otimes\mathcal{I}_{\rm BB'})\left(\proj{\rm GHZ}_{\rm XBB'}\right)\otimes\frac{\id_{\rm A'C'}}{4}.
\end{align}
It remains to compute the output of \CYthree{$\ket{{\rm GHZ}_{\rm XBB'}}$} after the cloning unitary, which reads

\begin{align}
\CYthree{\ket{{\rm GHZ}_{\rm XBB'}}}\otimes\ket{00}_{\rm CM}
&\mapsto\frac{1}{\sqrt{2}}\left[\sqrt{\frac{2}{3}}\ket{000}_{\rm BB'A}\ket{0}_{\rm C}\ket{1}_{\rm M} - \sqrt{\frac{1}{3}}\ket{00}_{\rm BB'}\frac{1}{\sqrt{2}}\left(\ket{10}+\ket{01}\right)_{{\rm AC}}\ket{0}_{\rm M}\right]\nonumber\\
&+\frac{1}{\sqrt{2}}\left[-\sqrt{\frac{2}{3}}\ket{111}_{\rm BB'A}\ket{1}_{\rm C}\ket{0}_{\rm M} + \sqrt{\frac{1}{3}}\ket{11}_{\rm BB'}\frac{1}{\sqrt{2}}\left(\ket{10}+\ket{01}\right)_{{\rm AC}}\ket{1}_{\rm M}\right]\nonumber\\
&=-\sqrt{\frac{1}{12}}\ket{001}_{\rm BB'A}\ket{0}_{\rm C}\ket{0}_{\rm M} + \sqrt{\frac{1}{3}}\left(\ket{000}_{\rm BB'A}+\frac{1}{2}\ket{111}_{\rm BB'A}\right)\ket{0}_{\rm C}\ket{1}_{\rm M}\nonumber\\
&-\sqrt{\frac{1}{3}}\left(\ket{111}_{\rm BB'A}+\frac{1}{2}\ket{000}_{\rm BB'A}\right)\ket{1}_{\rm C}\ket{0}_{\rm M} + \sqrt{\frac{1}{12}}\ket{110}_{\rm BB'A}\ket{1}_{\rm C}\ket{1}_{\rm M}.
\end{align}
Tracing away ${\rm M}$ and tensoring with $\frac{\id_{\rm A'C'}}{4}$, one can check that the resulting state is exactly the same as in Eq.~\eqref{Eq:NumericalGlobalChannel}, verifying the desired claim.

It is worth mentioning that the local noise model can also be found explicitly.
Let $\mathcal{N}_{\rm XB}$'s be the noise channels realizing the value $R(\{\mathcal{M}_{\rm AB},\mathcal{M}_{\rm CB}\}) = 0.75$; i.e., $\{0.75\mathcal{M}_{\rm XB} + 0.25\mathcal{N}_{\rm XB}\}_{\rm X = A,C}\in\mathfrak{C}$.
Then, numerically, we have
\begin{align}
\CYthree{\mathcal{N}_{\rm ABA'B'}^\mathcal{J}} = \frac{\id_{\rm A'}}{2}\otimes\left[\frac{1}{3}\left(\proj{001}+\proj{110}\right) + \frac{1}{6}\left(\proj{000}-\ket{000}\bra{111}-\ket{111}\bra{000}+\proj{111}\right)\right]_{\rm BB'A},
\end{align}
similar construction for \CYthree{$\mathcal{N}_{\rm CBC'B'}^\mathcal{J}$} can be obtained by replacing ${\rm A}$ by ${\rm C}$.
One can observe that
\begin{align}
\mathcal{N}_{\rm AB}(\cdot) = &\frac{2}{3}\left(\mathcal{I}_{\rm A}\otimes\mathcal{D}_{\rm B}\right)\circ{\rm CNOT}_{\rm AB}\left[\proj{1}_{\rm A}\otimes{\rm tr}_{\rm A}(\cdot)\right]+\frac{1}{3}{\rm CNOT}_{\rm AB}^-\left[\proj{0}_{\rm A}\otimes{\rm tr}_{\rm A}(\cdot)\right],
\end{align}
where $\mathcal{D}(\cdot)\coloneqq\proj{0}\cdot\proj{0} + \proj{1}\cdot\proj{1}$ and ${\rm CNOT}_{\rm XB}^-:\ket{i}_{\rm X}\otimes\ket{j}_{\rm B}\mapsto(-1)^j\ket{(i+j)_{\rm mod\;2}}_{\rm X}\otimes\ket{j}_{\rm B}$.
\CYthree{
As a last remark, note that, in a three-qubit setting, $\lambda\proj{\rm GHZ_{XYZ}} + (1-\lambda)\frac{\id_{\rm XYZ}}{8}$ is multipartite entangled if and only if $\lambda>\frac{1}{5}$~\cite{Dur2000,Schack2000}.
This means that the local channels become compatible before losing the ability to maintain multipartite entanglement shared between the local system and an external party (this can be seen by considering noises given by depolarizing channels whose Choi states are $\{\frac{\id_{\rm ABA'B'}}{16},\frac{\id_{\rm BCB'C'}}{16}\}$).
In this sense, channel incompatibility is a more fragile property than multipartite entanglement.
}

\section{CMP Is Irreducible to SMP}\label{App:CMP_Irreducible_to_SMP}
As mentioned in the main text, one may be tempted to conjecture the following:
\begin{center}
{\em Can the CMP be reduced to SMPs for the image states of the local channels?}
\end{center}
In other words, we are asking whether the channel compatibility of $\mathbfcal{E} = \{\mE_{\rm X|X'}\}_{\rm X|X'\in\Lambda}$ can be reduced to the state compatibility of the image states $\{\mE_{\rm X|X'}(\rho_{\rm X'})\}_{\rm X|X'\in\Lambda}$ for every input set consisting of compatible states $\{\rho_{\rm X'}\}$.
We can disprove this conjecture by constructing a family of counterexamples.
To this end, consider a bipartite state $\omega_{\rm XB'}$ satisfying
(1) $\omega_{\rm XB'}$ is not 2-extendible with respect to ${\rm B'}$~\cite{Kaur2021,Kaur2019};
(2) ${\rm tr}_{X}(\omega_{\rm XB'}) = \frac{\id_{\rm B'}}{d_{\rm B'}}$.
Then define a channel $\mathcal{W}_{\rm XB}:{\rm XB\to XB}$ with Choi state
\begin{align}
\mathcal{W}_{\rm XX'BB'}^\mathcal{J} \coloneqq \sigma_{\rm B}\otimes\omega_{\rm XB'}\otimes\frac{\id_{\rm X'}}{d_{\rm X'}},
\end{align}
where $\sigma_{\rm B}$ is a fixed (but arbitrarily chosen) state in ${\rm B}$.
Then we have the following result:

\begin{aproposition}\label{Example}
The channel $\mathcal{W}_{\rm XB}$ satisfies
\begin{enumerate}
\item $\{\mathcal{W}_{\rm AB},\mathcal{W}_{\rm CB}\}$ is locally compatible but incompatible.
\item $\{\mathcal{W}_{\rm AB}(\eta_{\rm AB}),\mathcal{W}_{\rm CB}(\eta_{\rm CB})\}$ is a compatible pair of states for every locally compatible pair of input states $\{\eta_{\rm AB},\eta_{\rm CB}\}$.
\end{enumerate}
\end{aproposition}
\begin{proof}
First, since locally in ${\rm B}$ both $\mathcal{W}_{\rm AB},\mathcal{W}_{\rm CB}$ are the state preparation channel of $\sigma_{\rm B}$, they are locally compatible.
Now, assume by contradiction that the pair $\{\mathcal{W}_{\rm AB},\mathcal{W}_{\rm CB}\}$ is compatible, then there exists a global channel whose Choi state contains a $2$-extension of the state $\omega_{\rm XB'}$, resulting in a contradiction.
The first claim is proved.
To see the second claim also holds, we note that for any given pair of locally compatible states $\{\eta_{\rm AB},\eta_{\rm CB}\}$ with marginal ${\rm tr}_{\rm A}(\eta_{\rm AB}) = {\rm tr}_{\rm C}(\eta_{\rm CB}) = \kappa$, we have that, for ${\rm X=A,C}$,
\begin{align}
\mathcal{W}_{\rm XB}(\eta_{\rm XB}) = \Omega_{\rm X|B}(\kappa)\otimes\sigma_{\rm B},
\end{align}
where $\Omega_{\rm X|B}$ is a channel from ${\rm B}$ to ${\rm X}$ with the Choi state $\Omega_{\rm XB'}^\mathcal{J} = \omega_{\rm XB'}$.
These two image states are always compatible with the tripartite state $\Omega_{\rm A|B}(\kappa)\otimes\sigma_{\rm B}\otimes\Omega_{\rm C|B}(\kappa)$.
This completes the proof.
\end{proof}
Note that the channel given in the main text
$
\mathcal{K}_{\rm XB}(\cdot)\coloneqq{\rm SWAP}[\proj{0}_{\rm X}\otimes{\rm tr}_{\rm X}(\cdot)]
$
has Choi state 
$
\mathcal{K}_{\rm XX'BB'}^\mathcal{J} = \proj{\Psi^+_{\rm XB'}}\otimes\proj{0}_{\rm B}\otimes\frac{\id_{\rm X'}}{d_{\rm X'}}.
$
Proposition~\ref{Example} implies that $\{\mathcal{K}_{\rm AB},\mathcal{K}_{\rm CB}\}$ is a particular case of a family of counterexamples to the above-mentioned conjecture.
For instance, one can consider two-qubit isotropic states $\rho_{\rm iso}(p)$~\cite{Horodecki1999,Horodecki1999-2} with $p$ strictly higher than the 2-extendible threshold $p=\frac{2}{3}$~\cite{Johnson2013} (see also Lemma 3 in Ref.~\cite{Kaur2021}).
This gives the following one-parameter family of counterexamples for $\frac{2}{3}<p\le1$:
\begin{align}
p\mathcal{K}_{\rm XB}(\cdot) + (1-p)\left(\frac{\id_{\rm X}}{2}\otimes\proj{0}_{\rm B}\right){\rm tr}(\cdot).
\end{align}
Our examples also show that channel incompatibility is {\em not equivalent to} violation of entanglement monogamy~\cite{Coffman2000} for the image states of local channels.
Finally, as another corollary, we have the following result about the resource theory of state unextendibility~\cite{Kaur2021,Kaur2019}:
\begin{acorollary}
The largest set of free operations of resource theory of state unextendibility is not the same with extendible channels.
\end{acorollary}
Hence, in the resource theory of state unextendibility, being resource non-generating is a condition strictly weaker than being extendible for a channel.

\section{Proofs of Main Theorems}\label{App:ProofsMainResults}
First, we recall the following decomposition from Proposition 7 in Ref.~\cite{Rosset2018} that will be used in the proof of Theorem~\ref{Coro}:
\begin{atheorem}\label{App:Rosset}{\rm\cite{Rosset2018}}
$W_{\rm AB}$ is a hermitian operator acting on a bipartite system ${\rm AB}$.
Then there exist states $\xi^{(i)}$ in ${\rm A}$ and $\rho^{(j)}$ in ${\rm B}$ and real numbers $\omega_{ij}$ such that
$
W_{\rm AB} = \sum_{i,j}\omega_{ij}\xi^{(i),T}\otimes\rho^{(j),T},
$
where $T$ is the transpose operator.
Moreover, the number of nonzero $\omega_{ij}$ is at most $d_{\rm min}^2 + 3$, where $d_{\rm min}$ is the smallest system dimension among ${\rm A,B}$.
\end{atheorem}

\subsection{Proof of Theorem~\ref{Coro}}\label{App:Proof-Coro}
\begin{proof}
We start from Eq.~\eqref{Eq:Result:CompatibilityProblemCondition} in the main text and apply Theorem~\ref{App:Rosset}.
For every $\mathbfcal{E}$ and $\{H_{\rm X|X'}\}_{{\rm X|X'\in\Lambda}}$ with $H_{\rm X|X'}^\dagger = H_{\rm X|X'}$ in the system ${\rm XX'}$, Theorem~\ref{App:Rosset} implies the existences of states $\xi^{(i)}_{{\rm X|X'}}$ in ${\rm X}$, $\rho^{(j)}_{{\rm X|X'}}$ in ${\rm X'}$, and real numbers $\{\omega^{(ij)}_{\rm X|X'}\}_{i,j=1}^{N_{\rm X|X'}}$ [we can choose $N_{\rm X|X'} \le \left(\min\{d_{\rm X},d_{\rm X'}\}\right)^2 + 3$ for every ${\rm X|X'\in\Lambda}$], such that 
\begin{align}\label{Eq:CompatibleSuffCondition}
\sum_{\rm X|X'\in\Lambda}{\rm tr}\left(H_{\rm X|X'}\mathcal{E}_{\rm XX'}^\mathcal{J}\right)&=\sum_{\rm X|X'\in\Lambda}\sum_{i,j=1}^{N_{\rm X|X'}}\omega^{(ij)}_{{\rm X|X'}}{\rm tr}\left[\left(\xi^{(i),T}_{{\rm X|X'}}\otimes\rho^{(j),T}_{{\rm X|X'}}\right)(\mE_{\rm X|X'}\otimes\mathcal{I}_{\rm X'})\left(\proj{\Psi_{\rm X'X'}^+}\right)\right]\nonumber\\
&=\sum_{\rm X|X'\in\Lambda}\sum_{i,j=1}^{N_{\rm X|X'}}\frac{\omega^{(ij)}_{{\rm X|X'}}{d_{\rm X'}}}{\rm tr}\left[\xi^{(i),T}_{{\rm X|X'}}\mE_{\rm X|X'}\left(\rho^{(j)}_{{\rm X|X'}}\right)\right]=\sum_{\rm X|X'\in\Lambda}\sum_{j=1}^{N_{\rm X|X'}}{\rm tr}\left[E^{(j)}_{{\rm X|X'}}\mE_{\rm X|X'}\left(\rho^{(j)}_{{\rm X|X'}}\right)\right],
\end{align}
where $
E^{(j)}_{{\rm X|X'}}\coloneqq\sum_{i=1}^{N_{\rm X|X'}}\frac{\omega^{(ij)}_{{\rm X|X'}}}{d_{\rm X'}}\xi^{(i),T}_{{\rm X|X'}}$,
which is a Hermitian operator in ${\rm X}$.
Using Eq.~\eqref{Eq:Result:CompatibilityProblemCondition}, and noticing that for every ${\rm X|X'}$ we have that $N_{\rm X|X'}\le\left(\min\{d_{\rm X},d_{\rm X'}\}\right)^2 + 3\le \left(\max_{\rm X|X'\in\Lambda}\{d_{\rm X},d_{\rm X'}\}\right)^2 +3$, the result follows by adding zero operators (which is Hermitian) to the set $\left\{E^{(j)}_{{\rm X|X'}}\right\}$.
\end{proof}

\subsection{Proof of Theorem~\ref{Result:EnsembleStateDiscrimination}}\label{App:Proof-Result:EnsembleStateDiscrimination}
Rather than proving the theorem directly, we show the following result that has Theorem~\ref{Result:EnsembleStateDiscrimination} as a direct corollary (note that, just like in the main paper, the subscripts of the following Hermitian operators and local states are showing the dependency on the output-input pair ${\rm X|X'}$ rather than the system they belong to):
\newpage
\begin{atheorem}\label{Result:Witness-DiscriminationTask}
For every $\mathbfcal{E}$, the following two statements are equivalent:
\begin{enumerate}
\item\label{Eq:ForCoro} For every \CY{${\rm X|X'}\in\Lambda$} there exist Hermitian operators \CY{$\left\{H^{(i)}_{{\rm X|X'}}\right\}_{i=1}^{N}$ in ${\rm X}$} and states \CY{$\left\{\rho^{(i)}_{{\rm X|X'}}\right\}_{i=1}^{N}$ in ${\rm X'}$}, where $N\in\mathbb{N}$ \CY{is independent of ${\rm X|X'}$}, such that
\begin{align}
\sum_{{\rm X|X'}\in\Lambda}\sum_{i=1}^{N}{\rm tr}\left[H^{(i)}_{{\rm X|X'}}\mE_{\rm X|X'}\left(\rho^{(i)}_{{\rm X|X'}}\right)\right] > \max_{\mathbfcal{L}\in\mathfrak{C}}\sum_{{\rm X|X'}\in\Lambda}\sum_{i=1}^{N}{\rm tr}\left[H^{(i)}_{{\rm X|X'}}\mathcal{L}_{\rm X|X'}\left(\rho^{(i)}_{{\rm X|X'}}\right)\right].
\end{align}
\item\label{Eq:GeneralDiscriminationTask} There exists a strictly positive $D$ such that
$
P(D,\mathbfcal{E}) > P_{\mathfrak{C}}(D).
$
\end{enumerate}
\end{atheorem}
\begin{proof}
First, we note that \CYtwo{Statement~\ref{Eq:GeneralDiscriminationTask}} implies \CYtwo{Statement~\ref{Eq:ForCoro}} with the Hermitian operator \CY{$H^{(i)}_{{\rm X|X'}}\coloneqq p_{\rm X|X'}q^{(i)}_{{\rm X|X'}}M^{(i)}_{{\rm X|X'}}\ge0$ \CYtwo{(one may also need to add zero operators to make the range of $i$'s independent of ${\rm X|X'}$)}.}
So it remains to show that \CYtwo{Statement~\ref{Eq:GeneralDiscriminationTask}} holds if \CYtwo{Statement~\ref{Eq:ForCoro}} is true.
As the first step, we note that \CYtwo{Statement~\ref{Eq:ForCoro}} holds if and only if 
\begin{align}
\CY{\sum_{{\rm X|X'}\in\Lambda}\sum_{i=1}^{N}{\rm tr}\left[\kappa\times(H^{(i)}_{{\rm X|X'}} + \Delta^{(i)}_{{\rm X|X'}}\id_{\rm X})\mE_{\rm X|X'}\left(\rho^{(i)}_{{\rm X|X'}}\right)\right]>\max_{\mathbfcal{L}\in\mathfrak{C}}\sum_{{\rm X|X'}\in\Lambda}\sum_{i=1}^{N}{\rm tr}\left[\kappa\times(H^{(i)}_{{\rm X|X'}} + \Delta^{(i)}_{{\rm X|X'}}\id_{\rm X})\mathcal{L}_{\rm X|X'}\left(\rho^{(i)}_{{\rm X|X'}}\right)\right],}
\end{align}
for every $\kappa>0$ and real numbers \CY{$\{\Delta^{(i)}_{{\rm X|X'}}\}$}.
To see that this is the case, it suffices to notice that, since \CY{$\mE_{\rm X|X'}$'s} are trace-preserving, \CY{$\sum_{i,{\rm X|X'}}{\rm tr}\left[\Delta^{(i)}_{{\rm X|X'}}\id_{\rm X}\times\mE_{\rm X|X'}\left(\rho^{(i)}_{{\rm X|X'}}\right)\right] = \sum_{i,{\rm X|X'}}\Delta^{(i)}_{{\rm X|X'}}$} is a fixed real number, and adding a fixed real number on both sides of \CY{the inequality in Statement~\ref{Eq:ForCoro}} preserves the inequality.
\CY{Let $Z^{(i)}_{{\rm X|X'}}\coloneqq\kappa\left(H^{(i)}_{{\rm X|X'}} + \Delta^{(i)}_{{\rm X|X'}}\id_{\rm X}\right)$, which is a Hermitian operator in ${\rm X}$ with dependency on \CYtwo{${\rm X|X'}$}}, we can choose $\kappa,\CY{\Delta^{(i)}_{{\rm X|X'}}}$ such that 
\CYtwo{
$
Z^{(i)}_{{\rm X|X'}}>0\;\;\forall\,i\;\&\;{\rm X|X'}$ and 
$
\sum_{i=1}^{N}Z^{(i)}_{{\rm X|X'}} < \id_{\rm X}\;\;\forall\,{\rm X|X'}.
$
Hence, for each ${\rm X|X'}$, $\{Z^{(i)}_{{\rm X|X'}}\}_{i=1}^{N}$} can be interpreted as part of a POVM.
Also, \CYtwo{Statement~\ref{Eq:ForCoro}} implies that there exists a set of states \CY{$\{\rho^{(i)}_{{\rm X|X'}}\}$} such that
\begin{align}\label{Eq:ComputationDetailDelta}
\CY{\sum_{{\rm X|X'}\in\Lambda}\sum_{i=1}^{N}{\rm tr}\left[Z^{(i)}_{{\rm X|X'}}\mE_{\rm X|X'}\left(\rho^{(i)}_{{\rm X|X'}}\right)\right]>\max_{\mathbfcal{L}\in\mathfrak{C}}\sum_{{\rm X|X'}\in\Lambda}\sum_{i=1}^{N}{\rm tr}\left[Z^{(i)}_{{\rm X|X'}}\mathcal{L}_{\rm X|X'}\left(\rho^{(i)}_{{\rm X|X'}}\right)\right].}
\end{align}
Now, consider the ensemble state discrimination task \CY{$D = \left(\{p_{\rm X|X'}\}, \{q^{(i)}_{{\rm X|X'}},\sigma^{(i)}_{{\rm X|X'}}\}, \{M^{(i)}_{{\rm X|X'}}\}\right)$} given by (\CY{what follows holds for every ${\rm X|X'}$; note again that the subscript now denotes the dependency on ${\rm X|X'}$ rather than the systems they live in}):
\begin{align}
&\CY{p_{\rm X|X'}} = \frac{1}{|\Lambda|};\\
&\CY{q^{(i)}_{{\rm X|X'}} = \frac{1-\epsilon}{N}\quad\;{\rm if}\;i=1,...,N\quad\&\quad q^{(N+1)}_{{\rm X|X'}} = \epsilon;}\\
&\CY{\sigma^{(i)}_{{\rm X|X'}} = \rho^{(i)}_{{\rm X|X'}}\quad\;{\rm if}\;i=1,...,N\quad\&\quad \sigma^{(N+1)}_{{\rm X|X'}} = \eta_{\rm X|X'};}\\
&\CY{M^{(i)}_{{\rm X|X'}} = Z^{(i)}_{{\rm X|X'}}\quad\;{\rm if}\;i=1,...,N\quad\&\quad M^{(N+1)}_{{\rm X|X'}} = \id_{\rm X} - \sum_{i=1}^{N}M^{(i)}_{{\rm X|X'}},}
\end{align}
where $\epsilon\in[0,1]$ is a real number whose range will be set later, and \CY{$\eta_{\rm X|X'}$'s are states in ${\rm X'}$} that can be chosen arbitrarily \CY{(but they still depend on ${\rm X|X'}$)}.
Then, according to the setting, \CY{$\{M^{(i)}_{{\rm X|X'}}\}_{i=1}^{N+1}$ is a POVM in the output system ${\rm X}$ for every ${\rm X|X'\in\Lambda}$}, which means $D$ is included in ensemble discrimination tasks.
Also note that $D$ is strictly positive once $0<\epsilon<1$.

For any set of channels \CY{$\mathbfcal{N}=\{\mathcal{N}_{\rm X|X'}\}_{\rm X|X'\in\Lambda}$}, we write
\CYtwo{
$
P(D,\mathbfcal{N}) = \wt{P}(\mathbfcal{N})+\epsilon\times\Gamma(\mathbfcal{N}),
$
}
where
\begin{align}
&\wt{P}(\mathbfcal{N})\coloneqq \frac{1}{N|\Lambda|}\CY{\sum_{{\rm X|X'\in\Lambda}}\sum_{i=1}^{N}{\rm tr}\left[Z^{(i)}_{{\rm X|X'}}\mathcal{N}_{\rm X|X'}\left(\rho^{(i)}_{{\rm X|X'}}\right)\right];}\\
&\Gamma(\mathbfcal{N})\coloneqq\CY{\frac{1}{|\Lambda|}\sum_{{\rm X|X'\in\Lambda}}\left({\rm tr}\left[\left(\id_{\rm X} - \sum_{i=1}^{N}Z^{(i)}_{{\rm X|X'}}\right)\mathcal{N}_{\rm X|X'}(\eta_{\rm X|X'})\right] - \frac{1}{N}\sum_{i=1}^{N}{\rm tr}\left[Z^{(i)}_{{\rm X|X'}}\mathcal{N}_{\rm X|X'}(\rho^{(i)}_{{\rm X|X'}})\right]\right).}
\end{align}
Note that Eq.~\eqref{Eq:ComputationDetailDelta} implies that $\wt{P}(\mathbfcal{E})=\Delta+\max_{\mathbfcal{L}\in\mathfrak{C}}\wt{P}(\mathbfcal{L})$ for some $\Delta>0$.
From here we conclude that:
\begin{align}
\max_{\mathbfcal{L}\in\mathfrak{C}}P(D,\mathbfcal{L}) &\le \max_{\mathbfcal{L}\in\mathfrak{C}}\wt{P}(\mathbfcal{L}) + \epsilon\times\max_{\mathbfcal{L}\in\mathfrak{C}}\Gamma(\mathbfcal{L})=\wt{P}(\mathbfcal{E})-\Delta+\epsilon\times\max_{\mathbfcal{L}\in\mathfrak{C}}\Gamma(\mathbfcal{L})=P(D,\mathbfcal{E})-\Delta+\epsilon\times\left[\max_{\mathbfcal{L}\in\mathfrak{C}}\Gamma(\mathbfcal{L}) - \Gamma(\mathbfcal{E})\right],
\end{align}
Set $\Delta'\coloneqq\max_{\mathbfcal{L}\in\mathfrak{C}}\Gamma(\mathbfcal{L}) - \Gamma(\mathbfcal{E})$, which is finite since $\Gamma$ is bounded in the set of all channels.
Then
\CYtwo{
$
P_\mathfrak{C}(D)\coloneqq\max_{\mathbfcal{L}\in\mathfrak{C}}P(D,\mathbfcal{L})\le P(D,\mathbfcal{E})-\Delta+\epsilon\Delta'.
$
}
If $\Delta'\le0$, then $P_\mathfrak{C}(D)<P(D,\mathbfcal{E})$ $\forall\epsilon\in[0,1]$.
If $\Delta'>0$, take $\epsilon<\min\left\{\frac{\Delta}{\Delta'},1\right\}$ which again gives $P_\mathfrak{C}(D)<P(D,\mathbfcal{E})$.
\end{proof}
\CYthree{
As a remark, such operational advantages can be extended to a general dynamical resource theory setups recently investigated in, e.g., Ref.~\cite{Theurer2019,Rosset2018,LiuWinter2019,LiuYuan2019,Takagi2019-3,Hsieh2020-2,Hsieh2020-3}.
See Ref.~\cite{HsiehDRT} for further details.
}


\begin{thebibliography}{99}
\bibitem{footnote:LC}Note that channels that are not locally compatible can still have non-zero $R$ value. For example, in a 3-qubit setting, consider $\mE_{\rm ABA'B'}^\mathcal{J} = \proj{00}_{\rm AB}\otimes\frac{\id_{\rm A'B'}}{4}$ and $\mE_{\rm CBC'B'}^\mathcal{J} = \proj{11}_{\rm CB}\otimes\frac{\id_{\rm C'B'}}{4}$.
Then they are not locally compatible in ${\rm B}$, while by considering the noise channels $\mathcal{N}_{\rm ABA'B'}^\mathcal{J} = \proj{01}_{\rm AB}\otimes\frac{\id_{\rm A'B'}}{4}$ and $\mathcal{N}_{\rm CBC'B'}^\mathcal{J} = \proj{10}_{\rm CB}\otimes\frac{\id_{\rm C'B'}}{4}$, one can check that $\frac{1}{2}\mE_{\rm XBX'B'}^\mathcal{J} + \frac{1}{2}\mathcal{N}_{\rm XBX'B'}^\mathcal{J}$ (${\rm X = A,C}$) will be the marginal of the global 3-qubit channel, whose Choi state is $\proj{0}_{\rm A}\otimes\frac{\id_{\rm B}}{2}\otimes\proj{1}_{\rm C}\otimes\frac{\id_{\rm A'B'C'}}{8}$.
Hence, they have $R\ge0.5$ even though they are not locally compatible.
\bibitem{footnote:strictly} This represents a task where all states $\rho^{(i)}_{{\rm X|X'}}$ have non-vanishing probability to be chosen and all outcomes have a chance of occurring.
\bibitem{footnote:EntMonogamy} Note that this example also shows that channel incompatibility is not equivalent to violating entanglement monogamy of image states.
\bibitem{footnote:ClassicalChannels} Write $P_{\rm AB|XY}$ as $\sum_{a,b,x,y}P(ab|xy)\proj{ab}\bra{xy}\cdot\ket{xy}$ with a probability distribution $\{P(ab|xy)\}$.
Similar for $P_{\rm BC|YZ}$ with $\{P(bc|yz)\}$ and their marignal $P_{\rm B|Y}$ with $\{P(b|y)\}$ satisfying $\sum_cP(bc|yz) = P(b|y) = \sum_aP(ab|xy)$ (see also~\cite{footnote:ClassicalChannelAsQuantumChannel}). 
$P_{\rm AB|XY},P_{\rm BC|YZ}$ are compatible via a global classical channel $P_{\rm ABC|XYZ}$ with $\{P(ab|xy)P(bc|yz)/P(b|y)\}$.
\bibitem{footnote:NSD} Leveraging the connection between compatibility and causality, $P_\mathfrak{C}(D)$ is the maximum success probability achievable when the local aiding dynamics $\mathbfcal{E}$ are the marginal of a global dynamics involving no communication from ${\rm S'\setminus X'}$ to ${\rm X}$ for every output-input pair ${\rm X|X'}$.
\bibitem{footnote1} To see the identity ${\rm tr}_{\rm A}(\rho_{\rm AB})^T = {\rm tr}_{\rm A}(\rho_{\rm AB}^T)$, we write $\rho_{\rm AB} = \sum_{ijkl}f_{ij|kl}\ket{i}\bra{j}_{\rm A}\otimes\ket{k}\bra{l}_{\rm B}$, where $\{\ket{i}_{\rm A}\}_i$ and $\{\ket{k}_{\rm B}\}_j$ are the bases for the definition of the transpose map $(\cdot)^T$. Then direct computation shows ${\rm tr}_{\rm A}(\rho_{\rm AB})^T = \left(\sum_n\sum_{ijkl}f_{ij|kl}\braket{n}{i}_{\rm A}\braket{j}{n}_{\rm A}\times\ket{k}\bra{l}_{\rm B}\right)^T = \sum_{nkl}f_{nn|kl}\ket{l}\bra{k}_{\rm B} = {\rm tr}_{\rm A}\left(\sum_{ijkl}f_{ij|kl}\ket{j}\bra{i}_{\rm A}\otimes\ket{l}\bra{k}_{\rm B}\right) = {\rm tr}_{\rm A}(\rho_{\rm AB}^T)$.
\end{thebibliography}
\end{document}